\documentclass[sigconf,nonacm]{acmart}

\usepackage[utf8]{inputenc}
\usepackage{amsmath}
\usepackage{amsfonts}
\usepackage{amsthm}
\usepackage{booktabs}
\usepackage{cleveref} 
\usepackage{dsfont}
\usepackage{enumitem}
\usepackage{mathtools}
\usepackage{siunitx}
\usepackage{thmtools} 
\usepackage{thm-restate} 

\newlist{inlinelist}{enumerate*}{1}
\setlist*[inlinelist]{label=(\roman*)}

\theoremstyle{plain}
\newtheorem{theorem}{Theorem}

\newtheorem{lemma}[theorem]{Lemma}
\newtheorem{assumption}{Assumption}
\crefname{assumption}{assumption}{assumptions}
\Crefname{Assumption}{Assumption}{Assumptions}
\theoremstyle{remark}

\DeclareMathOperator{\cov}{cov}
\DeclareMathOperator\supp{supp}
\DeclareMathOperator{\var}{var}
\newcommand{\EE}{\mathbb{E}}
\newcommand{\NN}{\mathbb{N}}
\newcommand{\PP}{\mathbb{P}}
\newcommand{\RR}{\mathbb{R}}
\newcommand{\Normal}{\mathcal{N}}
\newcommand{\bigO}{\mathcal{O}}
\newcommand{\regret}{\mathcal{R}}
\newcommand{\hG}{\widehat{G}}
\newcommand{\heta}{\widehat\eta}
\newcommand{\htheta}{\widehat\theta}
\newcommand{\J}[1]{J_{\mathrm{#1}}}

\AtBeginDocument{%
  }

\acmYear{2025}
\setcopyright{rightsretained}
\acmConference[KDD '25]{Proceedings of the 31st ACM SIGKDD Conference on Knowledge Discovery and Data Mining V.2}{August 3--7, 2025}{Toronto, ON, Canada}
\acmBooktitle{Proceedings of the 31st ACM SIGKDD Conference on Knowledge Discovery and Data Mining V.2 (KDD '25), August 3--7, 2025, Toronto, ON, Canada}
\acmDOI{10.1145/3711896.3736928}
\acmISBN{979-8-4007-1454-2/2025/08}

\begin{document}

\title{Empirical Bayes Selection for Value Maximization}

\author{Dominic Coey}
\authornote{Both authors contributed equally to this research.}
\email{coey@meta.com}
\orcid{0000-0002-3040-9662}
\author{Kenneth Hung}
\authornotemark[1]
\email{kenhung@meta.com}
\orcid{0000-0002-3911-182X}
\affiliation{%
  \department{Central Applied Science}
  \institution{Meta}
  \city{Menlo Park}
  \state{California}
  \country{USA}
}


\begin{abstract}
We study the problem of selecting the best $m$ units from a set of $n$ as $m / n \to \alpha \in (0, 1)$, where noisy, heteroskedastic measurements of the units' true values are available and the decision-maker wishes to maximize the aggregate true value of the units selected. Given a parametric prior distribution, the empirical Bayes decision rule incurs $\bigO_p(n^{-1})$ regret relative to the Bayesian oracle that knows the true prior. More generally, if the error in the estimated prior is of order $\bigO_p(r_n)$, regret is $\bigO_p(r_n^2)$. In this sense \emph{selection} of the best units is fundamentally easier than \emph{estimation} of their values. We show this regret bound is sharp in the parametric case, by giving an example in which it is attained. Using priors calibrated from a dataset of over four thousand internet experiments, we confirm that empirical Bayes methods perform well in detecting the best treatments with only a modest number of experiments.
\end{abstract}


\received{3 February 2025}
\received[revised]{23 May 2025}
\received[accepted]{23 May 2025}

\maketitle

\newcommand\kddavailabilityurl{https://doi.org/10.5281/zenodo.15538137}

\ifdefempty{\kddavailabilityurl}{}{
\begingroup\small\noindent\raggedright\textbf{KDD Availability Link:}\\
The source code of the simulations that generate the plots has been made publicly available at \url{\kddavailabilityurl}.
\endgroup
}

\section{Introduction}
\label{sec:intro}

In many important scientific and economic applications, decision-makers are presented with data on the performance of $n$ units, from which they must select a strict subset for further investigation or treatment. Examples include identifying the best teachers, hospitals, or athletes (\citet{brown2008season,dimick2010ranking,chetty2014measuringII}); genes associated with particular outcomes (\citet{efron2002empirical}); drug candidates (\citet{yu2020bayes}); or in the application of this paper, internet experiments. Each unit is associated with an unobserved true value, which is measured with heteroskedastic noise. The constraint that only $m < n$ units can be selected arises naturally when the decision-maker has limited resources to devote to the chosen units, and must restrict attention to the most promising candidates.

A desirable feature of a selection procedure is that the aggregate value of its selections is close to the maximum attainable value. Understanding how different selection procedures perform in this respect enables decision-makers to assess the quality of their decisions in the preceding applications. The empirical Bayes approach to this question involves estimating the unknown, prior distribution from which the true values are drawn, and selecting units with the highest estimated posterior means. We show that if the prior distribution is known to lie within some parametric class, empirical Bayes incurs regret of order $\bigO_p(n^{-1})$\footnote{$\bigO_p(\cdot)$ is the stochastic $\bigO$ notation commonly used in statistics, as defined in \citep{van2000asymptotic}. We write $X_n = \bigO_p(a_n)$ if $X_n / a_n$ is bounded in probability.} relative to the oracle Bayes decision rule in which the prior distribution is known.\footnote{We refer to the oracle Bayes decision rule rather than simply the Bayes decision rule throughout, to emphasize that the prior is unknown to the decision-maker.} This is faster than the usual $n^{-1/2}$ parametric rate of convergence from the central limit theorem. In this sense \emph{selection} is fundamentally easier than \emph{estimation}: picking a set of units with low regret is easier than pinning down the precise values of those units. This generalizes directly to the nonparametric case: regret converges to zero at the square of the rate that estimation error in the prior converges to zero.

The basic intuition for this result follows. First, mistakes, whether of inclusion or exclusion, are only likely to happen for those units whose true values are sufficiently close to a critical threshold. Units comfortably (or below) above that threshold will be correctly selected (or omitted) with high probability (i.e.\ with probability converging to $1$, abbreviated as w.h.p.). Second, even those mistakes cannot be too costly, as units incorrectly included or excluded are likely to be marginal—almost good enough to be selected, or almost bad enough to be omitted. The regret, which is the product of two terms corresponding to these factors, will therefore be second-order small. We show that our $\bigO_p(n^{-1})$ bound is sharp in the parametric case, by constructing an example in which regret is at least $C n^{-1}$ with non-vanishing probability for some positive constant $C$.

We illustrate this result with simulations based on internet experimentation data, where units correspond to experiments, and values to treatment effects. Heteroskedasticity arises because experiments vary in sample size. Technology companies may wish to identify a subset of best- or worst-performing experiments for further investigation, in the former case, as candidates to launch to production, or in the latter case, as candidates to stop early. When the follow-up investigation incurs some cost, it may only be feasible to select a strict subset of experiments for more analysis. In this application, as in others, the cost of mistakes depends on their magnitude—that is, on the difference between the aggregate value of the units selected and the units that should have been selected. We simulate true effects from a scale mixture of mean-zero Gaussians calibrated on this dataset, and evaluate the regret of the empirical Bayes approach for selecting the top $10\%$ of experiments. Consistent with our theoretical results, we find that regret is $\bigO_p(n^{-1})$. By comparison, identifying the set of the top $10\%$ experiments with all misclassifications being equally penalized regardless of their magnitude, or estimating the treatment effects of the selected experiments, or estimating the prior distribution itself, are all structurally harder problems, each of which only exhibits convergence at the usual parametric rate. 

\subsection{Related Work}

Our work builds on several large and active strands of the statistics and econometrics literature. Foundational work introducing and developing the empirical Bayes approach to statistics includes \citet{robbins1956empirical,kiefer1956consistency,robbins1964empirical,efron1973stein}. Applications of the selection problem have proliferated, as the problem of discerning between units which perform well or poorly on the basis of noisy, heteroskedastic measurements describes many real-world settings of interest. Previous work has studied identifying the best teachers (\citet{kane2008does,jacob2008can, harris2014skills,chetty2014measuringI,gilraine2020new}), the best medical facilities (\citet{thomas1994empirical,goldstein1996league,dimick2010ranking,hull2018estimating}), the best baseball players (\citet{efron1977stein,brown2008season}); differentially expressed genes (\citet{efron2002empirical,smyth2004linear}); promising drug candidates (\citet{yu2020bayes}); geographic areas associated with the greatest intergenerational mobility (\citet{bergman2019creating}) or mortality (\citet{marshall1991mapping}), and employers exhibiting the most evidence of discrimination (\citet{kline2021reasonable}). Internet experiments are particularly well-suited to empirical Bayes methods (\citet{deng2015objective,goldberg2017decision,azevedo2019empirical,coey2019improving,azevedo2020b,guo2020empirical}) as datasets are often large enough for accurate estimation of flexibly-specified priors, and the experiment-level sampling error is typically close to normally distributed. For these applications, the aggregate value of the selected units will often be an important component of the decision-maker's utility function. Our results provide theoretical and empirical support for selection based on such methods.

The literature on post-selection inference, including \citet{dahiya1974estimation,cohen1989two,gupta2002multiple,fithian2014optimal,hung2019rank,andrews2019inference,guo2021inference}, also studies selection problems, but differs from the present work in that its chief focus is estimating the values, differences or ranks of the selected units, rather than analyzing the regret associated with the selection. \citet{dahiya1974estimation,cohen1989two} provide estimates for the value of a selection unit. \citet{gupta2002multiple,fithian2014optimal,hung2019rank,andrews2019inference,guo2021inference} largely aim at frequentist inferences. While the notion of regret we consider averages over draws from the distribution of units' true values, an alternative line of inquiry beyond the scope of this paper would be to characterize admissible and minimax decision rules for the frequentist analog of the regret we define, considering the units' values as fixed constants.

\citet{gu2020invidious,mogstad2024inference} both study similar selection problems to the one we analyze. \citet{gu2020invidious} take an empirical Bayes approach to selecting the best units while controlling the marginal false discovery rate; \citet{mogstad2024inference} assert frequentist control over the familywise error rate, which amounts to a zero-one loss based on the correctness of the ranks. Both consider loss functions different from ours.
In their frameworks, mistakenly selecting or omitting any unit incurs a discrete cost, whereas in ours the cost of mistakenly selecting or omitting a marginal unit near the selection threshold is small. We view these as complementary perspectives. While in some decision problems mistakes may be undesirable per se, the aggregate performance of the selected units is typically still of interest. For teacher evaluations, for example, policy-makers may rightly be concerned with guarantees over the number of teachers who are incorrectly fired (\citet{mogstad2022comment}), but may also wish to understand how well their selection procedure is performing from the students' perspective, in terms of aggregate teacher ``value-added''. In other contexts, as in internet experimentation or drug discovery, the aggregate value of the selection is the primary concern, and it is harder to justify caring about the number of mistakes per se.

Closely related to our paper is \cite{chen2022empirical}, which notes the importance of empirical Bayes top-$m$ selection to various social science applications, and derives regret bounds for the problem. Those rate results are more favorable than the ones we present in cases where we can recover the posterior mean but not the prior fast. However in other cases, e.g.\ the parametric case, \cite{chen2022empirical} bound regret by a term converging slower than $n^{-1/2}$, while we prove $n^{-1}$ convergence and show that rate cannot in general be improved upon. We summarize this comparison in \Cref{tab:rates}. Furthermore, while the nonparametric rate of convergence of estimated priors to the truth is generally only logarithmic even for optimal procedures \citep{carroll1988optimal,fan1991optimal}, we may often observe a faster rate of convergence (e.g.\ see \Cref{fig:w1}) in practice, which our result translates to a tighter bound on the regret.

\begin{table*}[btp]
    \centering
    \begin{tabular}{l c c}
        \toprule
        Class of priors & Our bound & Bound from \cite{chen2022empirical} \\
        \midrule
        Parametric & $\bigO_p(n^{-1})$ & $\tilde\bigO(n^{-1/2})$ \\
        Finite support of unknown cardinality \cite{chen1995optimal} & $\bigO_p(n^{-1/2})$ & $\tilde\bigO(n^{-1/2})$ \\
        Density function has $k$ bounded derivative with bounded support \cite{carroll1988optimal} & $\bigO_p((\log n)^{-k})$ & $\tilde\bigO(n^{-1/2})$ \\ 
        \bottomrule
    \end{tabular}
    \caption{Summary of regret convergence rates given by our bound and the bound from \cite{chen2022empirical}. Note that $\tilde\bigO$ indicates some hidden $\log$-factors, and also the difference that \citet{chen2022empirical} bounds the expectation of regret when we bound it stochastically.}
    \label{tab:rates}
\end{table*}

The bound in \cite{chen2022empirical} goes through the mean squared error of the posterior means, which would be more pessimistic if the posterior mean of irrelevant items (e.g. those almost always or never selected) are hard to estimate. Meanwhile, our method relies on controlling the number of mistakes by estimating the distribution, which is likely more challenging than minimizing the mean squared error, leading to potential pessimism in a different way.

Our selection problem may remind readers of the multi-armed bandit literature that studies the problem of identifying the top $m$ arms with a certain probability, e.g.\ \citet{shang2020bai,chen2017top}. However to target the probability of correct selection is to consider discontinuous loss functions similar to ones in \citet{gu2020invidious,mogstad2024inference}. We also note that our selection problem is non-sequential which leads to new challenges, as poor choices of parameter in the prior cannot be overcome with additional samples in long run.

Finally, our work is related to the compound decision framework introduced in \citet{robbins1951asymptotically}, in which a simple decision is made for each unit and the overall loss is the sum of the loss from each individual decision. Convergence and rate results are available for empirical Bayes as applied to compound decision problems, e.g.\ \citet{hannan1965rate,van1977empirical,zhang1997empirical,gupta2005empirical,polyanskiy2021sharp}, but as \citet{weinstein2021permutation} observes this framework is rather restrictive and does not encompass the value maximization problem studied here of selecting the best $m$ of $n$ units. \citet{weinstein2021permutation} generalizes further to a class of simultaneous decision problems that are permutation invariant, which encompasses our problem of selecting $m$ units. However the optimal frequentist solution requires knowledge of the empirical c.d.f.\ (e.c.d.f.), or equivalently the order statistics, of the true effects $\mu_i$. Instead of studying the performance under pathological choices of $\mu_i$ as would be required in a minimax analysis, we take a more Bayesian approach to enable an analysis of regret without knowledge of the order statistics of the $\mu_i$'s.

\subsection{Our Contribution}

Our main contribution relative to this existing literature is to provide the first sharp regret bounds for parametric empirical Bayes selection, and to show how these same ideas extend to control regret in the nonparametric case. Our empirical work on internet experimentation complements this by verifying that regret is quantitatively modest in practice, given a reasonable number experiments of moderate precision. Together, our theoretical and empirical results suggest optimism for empirical Bayes approaches to selection when the decision-maker is primarily concerned with maximizing the aggregate value of the selected units, as opposed to correctly classifying the top units or estimating their values.

\section{The Top-\texorpdfstring{$m$}{m} Selection Problem}
\subsection{Setup}
\label{sec:setup}
There are $n$ units, each of which is associated with a unobserved true value $\mu_i \in \RR$ and an observed noise standard deviation $\sigma_i > 0$.\footnote{We assume $\sigma_i$ to be known, in line with past applications of empirical Bayes methods, e.g.\ \citet{weinstein2018group,guo2020empirical,deng2021post}.} The $\mu_i$ and $\sigma_i$ are distributed independently from each other and independently across experiments.\footnote{Independence of $\mu_i$ and $\sigma_i$ is a common maintained assumption in empirical Bayes methods, but may be unrealistic in some applications. \cite{chen2022empirical} treats this topic in detail.} Their unknown marginal distributions are denoted $G_0$ and $H_0$,
\[
    (\mu_i, \sigma_i) \sim G_0 \times H_0.
\]
We consider nondegenerate prior distributions $G$ belong to a potentially nonparametric family $\mathcal{M}$, that forms a metric space with $1$-Wasserstein distance $W_1(\cdot, \cdot)$. We assume the family is not misspecified, i.e. the family includes the truth $G_0$. For each unit $i$, the decision-maker observes a measurement $X_i \in \RR$, which is distributed as
\[
    X_i \mid \mu_i, \sigma_i \sim \Normal(\mu_i, \sigma_i^2).
\]
The decision-maker must choose $m$ units for some $m < n$. Their average utility given the index set of choices $J \subset \{1, 2, \ldots, n\}$ is $U(J) = \frac{1}{n}\sum_{i = 1}^n \mathds{1}(i \in J) \mu_i$. Let
\begin{equation}
\label{eq:post-mean}
    f_{G, \sigma_i}(X_i) = \frac{\int \mu \frac{1}{\sigma_i} \phi\left(\frac{X_i - \mu}{\sigma_i}\right) \,dG}{\int \frac{1}{\sigma_i} \phi\left(\frac{X_i - \mu}{\sigma_i}\right) \,dG}
\end{equation}
denote the posterior mean of $\mu_i$ given $X_i,\sigma_i$, assuming that the prior distribution of $\mu_i$ is $G$, where $\phi(\cdot)$ is the probability density function (p.d.f.) of a standard Gaussian. The true posterior mean is $f_{G_0, \sigma_i}(X_i)$. An estimator $\hG = \hG(X_1, \ldots, X_n, \allowbreak \sigma_1, \ldots, \sigma_n)$ of $G_0$ is available, where $\hG$ converges to $G_0$ at some rate $r_n$ in $1$-Wasserstein distance. It is used as the empirical Bayes prior, and in constructing posterior mean estimates, $f_{\hG, \sigma_i}(X_i)$. For simplicity we denote $f_{G_0, \sigma_i}(X_i)$ and $f_{\hG, \sigma_i}(X_i)$, the oracle and empirical Bayes posterior means for unit $i$, as $\theta_i$ and $\htheta_i$ respectively. We can view $\theta_i$ as being drawn i.i.d.\ from a distribution, which we denote $P$.

Given the observed data, an oracle Bayesian decision-maker maximizes expected utility (where the expectation is with respect to the posterior distribution over the unknown true values) by selecting the $m$ units with the highest values of $\theta_i$, breaking ties randomly. The empirical Bayes decision-maker mimics this rule, by selecting the $m$ units with the highest values of $\htheta_i$, breaking ties randomly. Letting $\J{EB}$ and $\J{Bayes}$ be the empirical Bayes and oracle Bayes choice sets, the regret
from empirical relative to oracle Bayes is
\begin{align}
    \regret &= \EE[U(\J{Bayes}) \mid X_1, \ldots, X_n, \sigma_1, \ldots, \sigma_n] - \\
    &\qquad \EE[U(\J{EB}) \mid X_1, \ldots, X_n, \sigma_1, \ldots, \sigma_n] \nonumber \\
    &= \frac{1}{n} \sum_{i = 1}^n (\mathds{1}(i \in \J{Bayes}) - \mathds{1}(i \in \J{EB})) \EE[\mu_i \mid X_i, \sigma_i] \nonumber \\
    &= \frac{1}{n} \sum_{i = 1}^n (\mathds{1}(i \in \J{Bayes}) - \mathds{1}(i \in \J{EB})) \theta_i, \label{eq:regret}
\end{align}
where $\mathds{1}(\cdot)$ is the indicator function.\footnote{We can also consider the loss $U(J) - U(\J{EB})$ for some other choice of benchmark $J$. However, other natural choices of $J$ may require oracle knowledge of the order statistics of the $\mu_i$'s, e.g.\ when $J$ is optimal among the class of permutation invariant choice sets (\citet{weinstein2021permutation}).} We aim to characterize how quickly $\regret$ converges to zero as $n \to \infty$ and $m / n \to \alpha \in (0, 1)$.\footnote{For simplicity of exposition, we only consider fixed $m$. We expect our main results to extend to the case where $m$ is allowed to be mildly data-driven, with $m/n \to \alpha$ in probability, as when selecting units with positive posterior means.} The proof that $\regret = \bigO_p(r_n^2)$ proceeds by bounding the regret $\regret$ by the product of two terms: the proportion of mistakes, and the maximum possible magnitude of the loss caused by a mistake. We show that each of these terms are of the same order as the estimation error in $\hG$, and consequently regret must be second-order small, i.e.\ $\bigO_p(r_n) \cdot \bigO_p(r_n) = \bigO_p(r_n^2)$.

Note that in the homoskedastic case when all variances are equal, $\sigma_1 = \cdots = \sigma_n$, the posterior mean of $\mu_i$ is monotone in $X_i$ for any choice of prior (\citet{efron2011tweedie,koenker2014convex}). Hence the oracle Bayes selection rule amounts to selecting the top-$m$ observations ordered by $X_i$, and this selection problem is trivial.

\subsection{Establishing a Convergence Bound}
\label{sec:convergence-bound}

To establish the convergence bound, we enlist \Cref{asm:g-rate,asm:compact-h}.

\begin{assumption}
\label{asm:g-rate}
    $W_1(\hG, G_0) = \bigO_p(r_n)$ for some sequence $(r_n)_{n \in \NN}$ with $r_n \ge n^{-1/2}$ for all $n$.
\end{assumption}

\begin{assumption}
\label{asm:compact-h}
    The support of the distribution $H_0$ of $\sigma_i$ is compact and bounded away from $0$.
\end{assumption}

\Cref{asm:g-rate} will be satisfied under mild conditions by the maximum likelihood estimator when $\mathcal{M}$ is parameterized by some finite-dimensional parameter $\eta$, with $r_n = n^{-1/2}$ (\citet[Theorem 9.14]{keener2010theoretical}).\footnote{The convergence rate holds for estimating $\eta$, but this translates to common families such as finite mixture families parameterized only by their weights and canonical exponential families with finite variance. For finite mixture families, note that the $1$-Wasserstein metric can be bounded by the product of the $L_1$-norm of the weight parameters and the maximum $1$-Wasserstein metric between any two mixture components. For exponential families, see \Cref{lma:exp-family} in \Cref{sec:real-analysis}.} This includes the commonly used ``normal-normal'' model, in which the prior is $\Normal(m_g, v_g)$ for unknown $m_g, v_g$ which are estimated by maximum likelihood. In the nonparametric case we will generally obtain slower rates of convergence, as allowed for by this assumption. For example, from the existing literature on convergence rates for deconvolution problems:
\begin{itemize}
    \item if the prior takes on some finite but unknown number of values, \citet{chen1995optimal} shows that the best possible convergence rate for estimating the prior in the $1$-Wasserstein metric\footnote{The $1$-Wasserstein metric is a natural choice here. We need a statistical distance that reflects the metric on the observation space, as our regret is tied to that metric as well. We also do not require the distributions to have the same support, or more precisely, be absolutely continuous with respect to $G_0$. Other common statistical distances such as total variation distance or Kullback--Leibler divergence do not meet these two desiderata.} is $n^{-1/4}$;
    \item if the prior has a density function with $k$ bounded derivatives, \citet{carroll1988optimal} shows that the fastest rate of convergence of any estimator of the prior is $(\log n)^{-k/2}$ in the $L_1$-norm of the p.d.f. Furthermore, if we assume the support of $G$ is bounded, the same rate of convergence applies to the $1$-Wasserstein metric.
\end{itemize}

\Cref{asm:compact-h} states that there are non-trivial upper and lower bounds on the precision with which the true values are measured, as would be the case in experiments with sample sizes bounded below and above.

Under \Cref{asm:g-rate,asm:compact-h}, our main result that $\regret = \bigO_p(r_n^2)$ follows. We give a brief overview of the proof strategy behind this theorem, before establishing supporting lemmas and giving the proof itself. Regret arises because our estimated posterior means, $\htheta_i$, are different from their oracle Bayes counterparts, $\theta_i$, and consequently the top units ranked by the former may differ from the top units ranked by the latter. The difference between $\htheta_i$ and $\theta_i$ is the error relative to oracle Bayes shrinkage for observation $i$. We show that regret can be bounded above by the product of the maximum magnitude of this shrinkage error from mistakes (whether of inclusion or exclusion) and the proportion of such mistakes. It suffices to show that each of these terms is $\bigO_p(r_n)$. Using the facts that the posterior mean function $f_{G,\sigma}(X_i)$ is sufficiently well-behaved around the $G_0$ when the observations $X_i$ belong to a compact set (\Cref{lma:lipschitz}), and that the $X_i$'s associated with all mistakes lie within a compact set w.h.p.\ (\Cref{lma:bounded-x}), we can show that the maximum magnitude of the shrinkage error from mistakes is bounded above by a constant times $W_1(\hG, G_0)$, and hence is $\bigO_p(r_n)$ by \Cref{asm:g-rate}. Next we argue that for any neighborhood around $P^{-1}(1 - \frac{m}{n})$ shrinking slower than $r_n$, the true values associated with mistakes will lie within that neighborhood w.h.p. This allows us to control the proportion of mistakes, and conclude that they are also $\bigO_p(r_n)$.


The following lemma is a key preliminary result, establishing continuity of both the posterior mean function $f_{G, \sigma}(X)$ and its inverse, and will be used to establish \Cref{lma:lipschitz,lma:bounded-x}. The existence of the inverse follows immediately from a classic result by \citet{efron2011tweedie}.

\begin{lemma}
\label{lma:continuity}
    Under \Cref{asm:compact-h}, the posterior mean function $f_{G, \sigma}(X)$ and its inverse $f^{-1}_{G, \sigma}(X)$ are both continuous in $(G, \sigma, X) \in \mathcal{M} \times \supp(H_0) \times \RR$.
\end{lemma}
\begin{proof}
    We first prove that $f_{G, \sigma}(X)$ is continuous in $(G, \sigma, X)$. From \eqref{eq:post-mean},
    \begin{align}
        \label{eq:f-explicit}
        f_{G, \sigma}(X) 
        &= \frac{\int \mu \phi\left(\frac{X - \mu}{\sigma}\right) \,dG}{\int \phi\left(\frac{X - \mu}{\sigma}\right) \,dG} \coloneqq \frac{h_1(G, \sigma, X)}{h_0(G, \sigma, X)},  
    \end{align}

    where $\phi(\cdot)$ is the p.d.f.\ of a standard Gaussian. As $h_0 > 0$, it suffices to show that $h_0$ and $h_1$ are continuous.

    Suppose we have a sequence $(G_k, \sigma_k, X_k) \to (G^*, \sigma^*, X^*)$ as $k \to \infty$. For $h_1$, we wish to show that
    \[
        \int \mu \phi\left(\frac{X_k - \mu}{\sigma_k}\right) \,dG_k \to \int \mu \phi\left(\frac{X^* - \mu}{\sigma^*}\right) \,dG^*.
    \]
    Note that the function sequence $\mu \phi\left(\frac{X_k - \mu}{\sigma_k}\right)$ converges uniformly to $\mu \phi\left(\frac{X^* - \mu}{\sigma^*}\right).$\footnote{For any $\varepsilon > 0$, there exists a compact interval $C_\varepsilon$ such that $\mu \phi\left(\frac{X_k - \mu}{\sigma_k}\right) < \varepsilon$ on $C_\varepsilon^c$. The function sequence itself is equicontinuous and converges pointwise, so it also converges uniformly within $C_\varepsilon$. Hence for any $\varepsilon > 0$ there is sufficiently large $k$ such that $\mu \phi\left(\frac{X_k - \mu}{\sigma_k}\right)$ is within $\varepsilon$ of $\mu \phi\left(\frac{X^* - \mu}{\sigma^*}\right)$ pointwise.} Hence for any $\varepsilon > 0$, when $k$ is sufficiently large, we have
    \begin{align}
        {}& \left| \int \mu \phi\left(\frac{X_k - \mu}{\sigma_k}\right) \,dG_k - \int \mu \phi\left(\frac{X^* - \mu}{\sigma^*}\right) \,dG_k \right| \nonumber \\
        \le{}& \sup_{\mu \in \RR} \left| \mu \phi\left(\frac{X_k - \mu}{\sigma_k}\right) - \int \mu \phi\left(\frac{X^* - \mu}{\sigma^*}\right) \right| \nonumber \\
        <{}& \varepsilon / 2. \label{eq:uniform-convergence}
    \end{align}
    Note also that $\mu \phi\left(\frac{X^* - \mu}{\sigma^*}\right)$ is Lipschitz. Therefore by Kantorovich--Rubinstein duality, when $k$ is sufficiently large, $W_1(G_k, G^*)$ is sufficiently small and
    \begin{equation}
    \label{eq:kr-duality}
        \left|\int \mu \phi\left(\frac{X^* - \mu}{\sigma^*}\right) \,dG_k - \int \mu \phi\left(\frac{X^* - \mu}{\sigma^*}\right) \,dG^*\right| < \varepsilon / 2.
    \end{equation}
    Summing up \eqref{eq:uniform-convergence} and \eqref{eq:kr-duality} yields the convergence of the numerator. The proof for $h_0$ is almost identical, as $\phi\left(\frac{X_k - \mu}{\sigma_k}\right)$ converges uniformly to $\phi\left(\frac{X^* - \mu}{\sigma^*}\right)$ and $\phi\left(\frac{X^* - \mu}{\sigma^*}\right)$ is Lipschitz. The continuity of $f_{G, \sigma}^{-1}(X)$ follows from the continuity of $f_{G, \sigma}(X)$ by \Cref{lma:shrinkage-increasing,lma:stack-ex}.
\end{proof}

The next lemma states that posterior mean function $f_{G, \sigma}(X)$ is locally Lipschitz around $G_0$, uniformly in $(\sigma, X) \in \supp(H_0) \times W$, for any compact $W$. This will be used in \Cref{thm:rn-squared} to bound the shrinkage error by a constant times the estimation error in the prior parameter, $\hG$.

\begin{lemma}
\label{lma:lipschitz}
    Suppose \Cref{asm:compact-h} holds. Then for any compact $W \subset \RR$, there exist positive constants $K$, $\delta$ such that for all $(G, \sigma, X) \in \mathcal{M} \times \supp(H_0) \times W$, we have $|f_{G, \sigma}(X) - f_{G_0, \sigma}(X)| \le K W_1(G, G_0)$ whenever $W_1(G, G_0) < \delta$.
\end{lemma}
\begin{proof}
    Using the same definition for $h_0$ and $h_1$ as in \eqref{eq:f-explicit}, we have
    \begin{align*}
        & |f_{G, \sigma}(X) - f_{G_0, \sigma}(X)| \\
        ={}& \left| \frac{h_1(G, \sigma, X)}{h_0(G, \sigma, X)} - \frac{h_1(G_0, \sigma, X)}{h_0(G_0, \sigma, X)} \right| \\
        \le{}& \frac{\left| h_1(G, \sigma, X) - h_1(G_0, \sigma, X) \right|}{h_0(G, \sigma, X)} \\
        & \quad + \frac{\left| h_1(G_0, \sigma, X) \right| \cdot \left| h_0(G_0, \sigma, X) - h_0(G, \sigma, X) \right|}{h_0(G, \sigma, X) h_0(G_0, \sigma, X)}.
    \end{align*}

    It remains to show the following claims for when $G$ is in a sufficiently small neighborhood of $G_0$:
    \begin{itemize}
        \item $h_0$ is bounded away from $0$: Since $h_0$ is continuous from the proof of \Cref{lma:continuity}, $h_0(G_0, \sigma, X)$ is strictly positive and $\supp(H_0) \times W$ is compact, for sufficiently small $\delta$, $h_0(G, \sigma, X)$ is bounded away from $0$ whenever $W_1(G, G_0) < \delta$.
        \item $h_1$ is Lipschitz in $G$ with a Lipschitz constant that does not depend on $\sigma$ or $X$: The integrand in $h_1$ is $\mu \phi\left(\frac{X - \mu}{\sigma}\right)$, a Lipschitz function in $\mu$. Since this Lipschitz constant is a continuous function in $\sigma, X$ and $\supp(H_0) \times W$ is compact, $\mu \phi\left(\frac{X - \mu}{\sigma}\right)$ is uniformly Lipschitz. The function $h_1$ is Lipschitz in $G$ again by Kantorovich--Rubinstein duality.
        \item $h_1(G_0, \sigma, X)$ is bounded: From the proof of \Cref{lma:continuity}, $h_1$ and thus $h_1(G_0, \cdot, \cdot)$ are bounded. So $h_1(G_0, \sigma, X)$ is bounded since $\supp(H_0) \times W$ is compact. \qedhere
    \end{itemize}
\end{proof}

We will apply \Cref{lma:lipschitz} on a specific $W$ which contains all of the observations corresponding to mistakes made by empirical Bayes selection w.h.p. This is the subject of the following lemma. We use $\triangle$ to denote symmetric difference, so $\J{Bayes} \triangle \J{EB}$ is the index set of all mistakes.

\begin{lemma}
\label{lma:bounded-x}
    If \Cref{asm:g-rate,asm:compact-h} hold, there exists a compact set $W$ such that $X_i \in W$ for all $i \in \J{Bayes} \triangle \J{EB}$ w.h.p.
\end{lemma}
\begin{proof}
    Oracle Bayes selection essentially thresholds on $\theta^*$, the $m$-th largest order statistic of the $\theta_i$'s. For any $c > 0$, this threshold lands in $(P^{-1}(1 - \frac{m}{n}) - c, P^{-1}(1 - \frac{m}{n}) + c)$ w.h.p.\ and hence $(P^{-1}(1 - \alpha)- c, P^{-1}(1 - \alpha) + c)$ w.h.p., by \citet[Corollary 21.5 and discussion thereof]{van2000asymptotic}. Let $V$ be a open ball of $G_0$ in $1$-Wasserstein whose radius is fixed but to be determined later. By \Cref{asm:g-rate}, $\hG$ lies in $V$ w.h.p. For any $i \notin \J{EB}$, under the high probability event $A_n$ defined as $A_n = \{\theta^* \in (P^{-1}(1 - \alpha) - c, P^{-1}(1 - \alpha) + c)\} \cap \{\hG \in V\}$,
    \begin{align}
        \htheta_i & \le \min_{i' \in \J{EB}} \htheta_{i'} \nonumber \\
        & = \min_{i' \in \J{EB}} f_{\hG, \sigma_{i'}} \circ f^{-1}_{G_0, \sigma_{i'}}(\theta_{i'}) \nonumber \\
        & \le \min_{i' \in \J{EB}} \max_{\sigma' \in \supp(H_0)} f_{\hG, \sigma'} \circ f^{-1}_{G_0, \sigma'}(\theta_{i'}) \label{eq:swap} \\
        & = \max_{\sigma' \in \supp(H_0)} f_{\hG, \sigma'} \circ f^{-1}_{G_0, \sigma'}\left(\min_{i' \in \J{EB}}\theta_{i'}\right) \nonumber \\
        & \le \max_{\sigma' \in \supp(H_0)} f_{\hG, \sigma'} \circ f^{-1}_{G_0, \sigma'}\left(\min_{i' \in \J{Bayes}}\theta_i\right) \\
        & \le \max_{\sigma' \in \supp(H_0)} f_{\hG, \sigma'} \circ f^{-1}_{G_0, \sigma'}(P^{-1}(1 - \alpha) + c) \label{eq:ub-for-exclusions} \\
        X_i & \le \max_{\sigma' \in \supp(H_0)} f^{-1}_{\hG, \sigma_i} \circ f_{\hG, \sigma'} \circ f^{-1}_{G_0, \sigma'}(P^{-1}(1 - \alpha) + c) \label{eq:unshrink} \\
        & \le \max_{\sigma', \sigma'' \in \supp(H_0)} f^{-1}_{\hG, \sigma''} \circ f_{\hG, \sigma'} \circ f^{-1}_{G_0, \sigma'}(P^{-1}(1 - \alpha) + c) \label{eq:ub-for-x}
    \end{align}
    By \Cref{lma:shrinkage-increasing,lma:stack-ex}, the function $\max_{\sigma' \in \supp(H_0)} f_{\hG, \sigma'} \circ f^{-1}_{\hG, \sigma'}$ is a strictly increasing function, so we can move the minimum inside in \eqref{eq:swap}. Also applying $f^{-1}_{\hG, \sigma_i}$ to both sides of \eqref{eq:ub-for-exclusions} yields \eqref{eq:unshrink}. Note that the maximand in \eqref{eq:ub-for-x} is a composition of functions in $\hG, \sigma', \sigma''$ that are continuous by \Cref{lma:continuity}, hence also a continuous function itself. By maximum theorem and the compactness of $\supp(H_0)$, \eqref{eq:ub-for-x} is continuous in $\hG$ and locally bounded. In other words, for a sufficiently small open ball in $1$-Wasserstein, $V$, centered at $G_0$, \eqref{eq:ub-for-x} is bounded by some constant.

    On the other hand, for any $i \in \J{Bayes}$, under the event $A_n$,
    \begin{align*}
        X_i &= f^{-1}_{G_0, \sigma_i}(\theta_i) \\ 
            &\ge f^{-1}_{G_0, \sigma_i}(P^{-1}(1 - \alpha) - c) \\ &\ge \min_{\sigma' \in \supp(H_0)} f^{-1}_{G_0, \sigma'}(P^{-1}(1 - \alpha) - c).
    \end{align*}

    Together, there exists a constant bounded interval that contains all $i$ in $\J{Bayes} \setminus \J{EB}$ under the event $A_n$. Likewise, there is also a constant bounded interval contains all $i$ in $\J{EB} \setminus \J{Bayes}$ under the event $A_n$. Taking $W$ to be the union of these two intervals completes the proof.
\end{proof}

With these preliminaries we can prove our main result.

\begin{restatable}{theorem}{rnsquared}
\label{thm:rn-squared}
    If \Cref{asm:g-rate,asm:compact-h} hold, then $\regret = \bigO_p(r_n^2)$, the square of the rate of convergence for estimating the prior.
\end{restatable}
\begin{proof}
    We first decompose an upper bound for $\regret$ into two components.
    \begin{align}
        \regret &\le \frac{1}{n}\sum_{i = 1}^n (\mathds{1}(i \in \J{Bayes}) - \mathds{1}(i \in \J{EB})) \theta_i \nonumber \\
        & \qquad - \frac{1}{n}\sum_{i = 1}^n (\mathds{1}(i \in \J{Bayes}) - \mathds{1}(i \in \J{EB})) \htheta_i \label{eq:eb-ranks-by-hat} \\
        &= \frac{1}{n}\sum_{i = 1}^n (\mathds{1}(i \in \J{Bayes}) - \mathds{1}(i \in \J{EB})) (\theta_i - \htheta_i) \nonumber \\
        &\le \frac{1}{n} \left(\#(\J{Bayes} \setminus \J{EB}) + \#(\J{EB} \setminus \J{Bayes})\right) \nonumber \\
        & \qquad \cdot \max_{i \in \J{Bayes} \triangle \J{EB}} |\theta_i - \htheta_i| \label{eq:sym-mistakes} \\
        &= 2 \cdot \underbrace{\frac{1}{n} \#(\J{Bayes} \setminus \J{EB})}_{\text{proportion of mistakes}} \cdot \underbrace{\max_{i \in \J{Bayes} \triangle \J{EB}} |\theta_i - \htheta_i|}_{\text{max magnitude of shrinkage error}} \label{eq:two-parts},
    \end{align}
    where \eqref{eq:eb-ranks-by-hat} follows from the fact that $\J{EB}$ is the set of indices of the $m$ largest $\htheta_i$'s. In \eqref{eq:sym-mistakes}, since $\#\J{Bayes} = \#\J{EB}$, we have $\#(\J{Bayes} \setminus \J{EB}) = \#(\J{EB} \setminus \J{Bayes})$. From \eqref{eq:two-parts}, it suffices to bound the proportion of mistakes and the maximum magnitude of shrinkage error.

    We start by bounding the latter term. We denote the event that the observations associated with all mistakes belong in the set $W$ from \Cref{lma:bounded-x} and $\hG$ belongs to a neighborhood $V$ around $G_0$ by $A_n = \{X_i \in W \text { for all } i \in \J{Bayes} \triangle \J{EB}\} \cap \{\hG \in V\}$. By \Cref{lma:bounded-x}, $\PP(A_n) \to 1$, and under this high probability event, we have
    \begin{align}
        \max_{i \in \J{Bayes} \triangle \J{EB}} |\theta_i - \htheta_i| &= \max_{i \in \J{Bayes} \triangle \J{EB}} |f_{G_0, \sigma_i}(X_i) - f_{\hG, \sigma_i}(X_i)| \nonumber \\
        &\le \max_{\substack{X \in W \\ \sigma \in \supp(H_0)}} |f_{G_0, \sigma}(X) - f_{\hG, \sigma}(X)| \nonumber \\
        &\le K W_1(G_0, \hG). \label{eq:max-grad}
    \end{align}
    \eqref{eq:max-grad} follows from \Cref{lma:lipschitz}, and is $\bigO_p(r_n)$ by \Cref{asm:g-rate}. Consequently
    \begin{equation}
    \label{eq:max-diff}
        \max_{i \in \J{Bayes} \triangle \J{EB}} |\theta_i - \htheta_i| = \bigO_p(r_n).
    \end{equation}

    Next we bound the proportion of mistakes. Let $\theta^*$ denote the $m$-th largest order statistic of the $\theta_i$'s, and $\theta^{**} = P^{-1}(1 - m/n)$ the $(1 - m/n)^\textrm{th}$ quantile of $P$.
    
    For any nondecreasing sequence $(b_n)_{n \in \NN}$ with $\lim_{n \to \infty} b_n = \infty$, define the event $B_n$ as
    \[
    B_n = \left\{\max_{i \in \J{Bayes} \setminus \J{EB}} \left|\theta_i - \theta^{**}\right| \le b_n r_n \right\}.
    \]
    Then
    \begin{align*}
        {}& \frac{1}{n} \#(\J{Bayes} \setminus \J{EB}) \\
        \le{}& \mathds{1}(B_n^c) +  \mathds{1}(B_n) \frac{1}{n}\#\left\{i: \left|\theta_i - \theta^{**}\right| \le b_n r_n\right\}.
    \end{align*}
        
    We first argue that $\PP(B_n^c) \rightarrow 0$ and subsequently that $\frac{1}{n}\#\{i: |\theta_i - P^{-1}(1 - \frac{m}{n})| \le b_n r_n\} = \bigO_p(b_n r_n)$, giving $\frac{1}{n} \#(\J{Bayes} \setminus \J{EB}) = \bigO_p(b_n r_n)$. Since $(b_n)_{n\in\NN}$ was an arbitrary nondecreasing sequence converging to infinity, by \Cref{lma:op} in \Cref{sec:real-analysis} this implies $\frac{1}{n} \#(\J{Bayes} \setminus \J{EB}) = \bigO_p(r_n)$.

    For each $i$ in $\J{Bayes} \setminus \J{EB}$, empirical Bayes selection must have excluded it because some other shrinkage estimate was larger, i.e.\ $\htheta_i \le \htheta_{i'}$ for some $i' \in \J{EB} \setminus \J{Bayes}$. Hence for each $i \in \J{Bayes} \setminus \J{EB}$, there is a $i' \in \J{EB} \setminus \J{Bayes}$ such that $\theta^* \le \theta_i 
        \le \theta_{i'} + 2 \max_{i \in \J{Bayes} \triangle \J{EB}} |\theta_i - \htheta_i|
        \le \theta^* + 2 \max_{i \in \J{Bayes} \triangle \J{EB}} |\theta_i - \htheta_i|$
    and so
    \begin{equation}
    \label{eq:bound-diff}
        \max_{i \in \J{Bayes} \setminus \J{EB}} | \theta_i - \theta^*| \le 2 \max_{i \in \J{Bayes} \triangle \J{EB}} |\theta_i - \htheta_i|.
    \end{equation}
    From the triangle inequality and union bound, we have
    \begin{align*}
    & \PP(B_n^c) \le \PP\left(r_n^{-1}\max_{i \in \J{Bayes} \setminus \J{EB}} | \theta_i - \theta^*| > b_n/2 \right) + \\
    & \quad \PP\left(r_n^{-1} \left|\theta^* - \theta^{**}\right| > b_n/2\right).
    \end{align*}
    By \eqref{eq:max-diff} and \eqref{eq:bound-diff}, $r_n^{-1}\max_{i \in \J{Bayes} \setminus \J{EB}} | \theta_i - \theta^*| = \bigO_p(1)$. By \Cref{lma:p-diff}, given standard results on the convergence of sample quantiles (\citet[Corollary 21.5]{van2000asymptotic}), we have $r_n^{-1} | \theta^* - P^{-1}(1 - \frac{m}{n})| = \bigO_p(r_n^{-1} n^{-1/2})$. As $b_n \rightarrow \infty$, $\PP(B_n^c) \to 0$.

    The probability of $\theta_i$ that falls in $(P^{-1}(1 - \frac{m}{n}) - b_n r_n, P^{-1}(1 - \frac{m}{n}) + b_n r_n)$ is no greater than
    $
    P\left(\theta^{**} + b_n r_n\right) - P\left(\theta^{**} - b_n r_n\right) = \bigO(b_n r_n)
    $
    by the continuous differentiability of $P$ from \Cref{lma:p-diff}. So by Chebyshev's inequality, the proportion $\frac{1}{n}\#\{i: |\theta_i - P^{-1}(1 - \frac{m}{n})| < b_n r_n\}$ is $\bigO_p(b_n r_n)$, and by arbitrariness of $b_n$, it must also be $\bigO_p(r_n)$.

    Because the first part and second parts of \eqref{eq:two-parts} are both $\bigO_p(r_n)$, it follows that the regret $\regret$ is $\bigO_p(r_n^2)$.
\end{proof}

The two main estimation approaches for empirical Bayes are {\em $f$-modeling}, in which a model is specified for the observed outcomes, and {\em $g$-modeling}, in which a model is specified for the unobserved prior (\citet{efron2014two}). This theorem is consistent with either estimation approach. In the $f$-modelling case, if the estimated distribution for outcomes is consistent with some prior distribution for true effects, i.e.\ falls in the class characterized by \citet{guo2020empirical}, we can think of $\hG$ as the prior implicitly specified by deconvolving the estimated observation distribution. For $g$-modelling, we can interpret $\hG$ directly as the model specified for the unobserved prior.

The bound is also sharp when $r_n = n^{-1/2}$, as shown by our example in \Cref{sec:sharpness}.

\section{Sharpness of the Convergence Bound in the Parametric Case}
\label{sec:sharpness}

We provide an example where the regret satisfies $\regret \ge C n^{-1}$ with non-vanishing probability for some positive constant $C$. Let the location family $G(\eta) = \Normal(\eta, 1)$ be the model for the prior, where the scalar location parameter $\eta$ is estimated by maximum likelihood. In our example, the truth is $\eta_0 = 0$. We assume the standard deviation of the noise term is drawn i.i.d.\ from
\[
    \sigma_i = \begin{cases}
        1 & \text{with probability $1/2$, and} \\
        2 & \text{with probability $1/2$;}
    \end{cases}
\]
and we will select $m = \lfloor \alpha n \rfloor$ units.

The maximum likelihood estimator $\heta$ converges to $\eta_0$ at rate $n^{-1/2}$. The oracle Bayes shrunken estimate of the posterior mean is $\theta_i = \frac{1}{\sigma_i^2 + 1} X_i$ and the empirical Bayes estimate is $\htheta_i = \frac{1}{\sigma_i^2 + 1} X_i + \frac{\sigma_i^2}{\sigma_i^2 + 1} \heta$. In particular, the magnitude of $\htheta_i - \theta_i = \frac{\sigma_i^2}{\sigma_i^2 + 1} \heta$ increases with $\sigma_i$. In our setting $\theta_i$ and $\sigma_i$ are measurable with respect to the Lebesgue measure and the counting measure, respectively. The density of $(\theta_i, \sigma_i)$ with respect to the product measure is then $\frac{1}{2} \cdot \sqrt{2} \phi(\sqrt{2}\theta_i)$ for $\sigma_i = 1$ and $\frac{1}{2} \cdot \sqrt{5}\phi(\sqrt{5}\theta_i)$ for $\sigma_i = 2$. If we condition on $\theta^*$ and $\theta_i < \theta^*$, then $\theta_i$ are i.i.d. In fact $(\theta_i, \sigma_i) \mid \theta^*, \theta_i < \theta^*$ is i.i.d.\ with density
\begin{equation}
\label{eq:trunc-mix-low}
    \begin{cases}
        \frac{\sqrt{2}\phi(\sqrt{2}\theta_i)}{\Phi(\sqrt{2}\theta^*) + \Phi(\sqrt{5}\theta^*)} \mathds{1}(\theta_i < \theta^*) & \text{for } \sigma_i = 1, \\
        \frac{\sqrt{5}\phi(\sqrt{5}\theta_i)}{\Phi(\sqrt{2}\theta^*) + \Phi(\sqrt{5}\theta^*)} \mathds{1}(\theta_i < \theta^*) & \text{for } \sigma_i = 2,
    \end{cases}
\end{equation}
with respect to the product measure, where $\Phi(\cdot)$ is the c.d.f.\ of a standard Gaussian. Likewise, $(\theta_i, \sigma_i) \mid \theta^*, \theta_i > \theta^*$ is i.i.d.\ with density
\begin{equation}
\label{eq:trunc-mix-high}
    \begin{cases}
        \frac{\sqrt{2}\phi(\sqrt{2}\theta_i)}{1 - \Phi(\sqrt{2}\theta^*) + 1 - \Phi(\sqrt{5}\theta^*)} \mathds{1}(\theta_i > \theta^*) & \text{for } \sigma_i = 1, \\
        \frac{\sqrt{5}\phi(\sqrt{5}\theta_i)}{1 - \Phi(\sqrt{2}\theta^*) + 1 - \Phi(\sqrt{5}\theta^*)} \mathds{1}(\theta_i > \theta^*) & \text{for } \sigma_i = 2.
    \end{cases}
\end{equation}

Consider a compact interval $[\underline{a}, \bar{a}]$ that contains $P^{-1}(1 - \alpha)$ in its interior. Since $\theta^*$ converges to $P^{-1}(1 - \alpha)$ at rate $n^{-1/2}$, the event $A_n$ where the interval $(\theta^* - c n^{-1/2}, \theta^* + c n^{-1/2})$ is a subset of $[\underline{a}, \bar{a}]$ happens w.h.p.\ for any positive constant $c > 0$. For any such $\theta^*$, the density in \eqref{eq:trunc-mix-low} over $(\theta^* - cn^{-1/2}, \theta^*)$ and density in \eqref{eq:trunc-mix-high} over $(\theta^*, \theta^* + cn^{-1/2})$ are in some strictly positive bounded interval $[\underline{b}, \bar{b}]$ that does not depend on the value of $\theta^*$.

There are three sets of units of interest:
\begin{itemize}
    \item $K_n = \{i: \theta_i \in (\theta^*, \theta^* + c n^{-1/2}) \text{ and } \sigma_i = 1\}$,
    \item $L_n = \{i: \theta_i \in (\theta^* - d n^{-1/2}, \theta^*) \text{ and } \sigma_i = 2\}$, and 
    \item $M_n = \{i: \theta_i \in (\theta^* - d n^{-1/2}, \theta^* - \frac{1}{2} d n^{-1/2}) \text { and } \sigma_i = 2\}$,
\end{itemize}
where $c, d$ are positive constants to be chosen. There are $\lfloor \alpha n \rfloor - 1$ realizations of $\theta_i$ greater than $\theta^*$ and $n - \lfloor \alpha n \rfloor$ realizations of $\theta_i$ smaller than $\theta^*$. Conditional on $\theta^*$, the cardinalities are consequently binomially distributed with
\begin{align*}
    \# K_n &\sim \text{Binomial}(\lfloor \alpha n \rfloor - 1, p_{K_n}(\theta^*)), \quad \text{where $p_{K_n}(\theta^*) > \underline{b} c n^{-1/2}$,} \\
    \# L_n &\sim \text{Binomial}(n - \lfloor \alpha n \rfloor, p_{L_n}(\theta^*)), \\
    & \qquad \text{where $\underline{b} d n^{-1/2} < p_{L_n}(\theta^*) < \bar{b} d n^{-1/2}$}, \\
    \# M_n &\sim \text{Binomial}(n - \lfloor \alpha n \rfloor, p_{M_n}(\theta^*)), \\
    & \qquad \text{where $p_{M_n}(\theta^*) > \frac{1}{2} \underline{b} d n^{-1/2}$.}
\end{align*}
Marginalizing over the event $A_n$ gives the same observation but removes the dependence on $\theta^*$. Hence for some constants $c_K, c_L, c_M > 0$, we have w.h.p.\ $\#K_n \ge c_K n^{1/2}$, $\#L_n \ge c_L n^{1/2}$, and $\#M_n \ge c_M n^{1/2}$. Furthermore $d > 0$ can be chosen sufficiently small such that $\# K_n > \# L_n$ w.h.p.

The rest of the argument focuses on the event where $\heta > \frac{10}{3} (c + d) n^{-1/2}$, which occurs with non-vanishing probability. Under this event, since $\heta > 0$, empirical Bayes selection will only mistakenly select units with $\sigma_i = 2$ in place of other units with $\sigma_i = 1$. In particular, for any $i$ in $K_n$ and $i'$ in $L_n$, we have $\theta_i > \theta_{i'}$ but
\[
    \htheta_i < \theta^* + cn^{-1/2} + \frac{1}{2}\heta < \theta^* - dn^{-1/2} + \frac{4}{5}\heta < \htheta_{i'}.
\]
So w.h.p.\ at least $\min(\# K_n, \# L_n) = \#L_n$ mistakes were made. In fact since $L_n$ consists of units immediately smaller than $\theta^*$ and the relative ordering of all units with $\sigma_i = 2$ does not change, all of $\# L_n$ will be mistakenly selected. This incurs a regret of at least
\[
    \frac{1}{n} \sum_{i \in L_n} (\theta^* - \theta_i) \ge \frac{1}{n} \sum_{i \in M_n} (\theta^* - \theta_i) \ge \frac{1}{n} \# M_n \cdot \frac{1}{2} d n^{-1/2} \ge \frac{1}{2} c_M d n^{-1},
\]
with high probability.

\section{Top-\texorpdfstring{$m$}{m} Selection in Simulation}
\label{sec:empirics}

We illustrate \Cref{thm:rn-squared} with a realistic simulation, based on the Upworthy dataset of internet experiments conducted between 2013 and 2015.\footnote{From the publicly accessible Upworthy Research Archive (\citet{matias2021upworthy}) which is downloadable at \url{https://osf.io/jd64p/}.} The dataset contains a list of experiments, along with effect sizes and standard errors. For the prior $G_0$, we fit a normal scale mixture with fixed components, parameterized only by the weights. The data and modelling details are described in \Cref{sec:sim-details}, and the notebook to reproduce the simulations and figures is available as an artifact\footnote{Also available on \url{https://github.com/facebookresearch/eb-selection}}.

We simulate a variety of signal-to-ratio regimes, and choices of family for the prior. In increasing order of flexibility, these are:
\begin{inlinelist}
    \item the family of normal priors,
    \item the family of scale mixtures of normals, and
    \item the family of all distributions.
\end{inlinelist}
The priors for these cases are estimated using the {\tt ebnm} R package (\citet{willwerscheid2021ebnm}). In particular, the normal scale mixture is estimated using adaptive shrinkage as in \citet{stephens2016fdr}, and the fully nonparametric case is estimated by nonparametric maximum likelihood estimator (NPMLE) (\citet{kiefer1956consistency}). Henceforth we refer to these three estimators as EB-NN, EB-NSM, and EB-NPMLE. This enables comparison of the performance of empirical Bayes methods under misspecification (when the restrictive EB-NN estimator is used), under a parsimonious and well-specified model (EB-NSM), and under a highly flexibly and well-specified model (EB-NPMLE). For the distribution $H_0$, we use the empirical distribution of standard errors in the dataset.

Top-$m$ selection here corresponds to selecting a subset of experiments, given a constraint on the subset size. We pick $m = \lfloor 0.1 n \rfloor$ and vary $n$, the number of simulated experiments,  showing the distribution of regret for each choice of $n$. For each $n$ we run $\num{1000}$ iterations of the selection simulation. In each iteration, we
\begin{enumerate}
    \item independently draw $n$ true treatment effects $\mu_i \sim G_0$ and noise standard deviations $\sigma_i \sim H_0$;
    \item generate the $n$ observations $X_i$, where $X_i \mid \mu_i, \sigma_i \sim \Normal(\mu_i, \sigma_i^2)$;
    \item fit three models for the prior distribution of treatment effects $\mu_i$: EB-NN, EB-NSM, and EB-NPMLE;
    \item compute the choice sets $\J{Bayes}$, $\J{EB-NN}$, $\J{EB-NSM}$, $\J{EB-NPMLE}$, and $\J{UN}$ corresponding to the oracle Bayes posterior mean estimators, the three empirical Bayes posterior mean estimators, and the unshrunk $X_i$;
    \item compute the regret relative to oracle Bayes selections, $\regret_{M} = \frac{1}{n} \sum_{i = 1}^n (\mathds{1}(i \in \J{Bayes}) - \mathds{1}(i \in \J{M}))\theta_i$
    for $M = \text{EB-NN},\allowbreak \text{EB-NSM},\allowbreak \text{EB-NPMLE},\allowbreak \text{UN}$.
\end{enumerate}
To assess performance in lower signal-to-noise regimes, we repeat this exercise with varying levels of sampling error. We use standard errors 1, 2 and 4 times greater than the baseline standard errors, corresponding to signal-to-noise ratios of roughly $1.3$, $0.7$ and $0.3$.

The normal scale mixture is a parametric model once the number of components and the scale parameters are fixed. As {\tt ebnm} fits by maximizing the likelihood, the remaining parameters—the weights—converge at $\bigO_p(n^{-1/2})$. Hence by \Cref{thm:rn-squared}, $\regret_\text{EB-NSM}$ is $\bigO_p(n^{-1})$. We have no such guarantees for $\regret_\text{EB-NN}$ or $\regret_\text{EB-UN}$, corresponding to the misspecified normal prior and the ``naive'' choice which selects the units with the largest $X_i$'s. EB-NPMLE is highly flexible and not misspecified, although its guaranteed convergence rate is very slow \citep{soloff2024multivariate}.

\Cref{fig:regret} shows regret as a function of the number of experiments, for each selection method and each value of the noise multiplier. As $n$ increases, the mean and $99$th percentile across simulations of $\regret_\text{EB-NSM}$ and $\regret_\text{EB-NPMLE}$ both exhibit declines consistent with $n^{-1}$ convergence, although the regret associated with the latter is larger, suggesting the NPMLE model incurs a cost from its greater flexibility. With just $\num{1000}$ experiments, the regret of the EB-NSM approach can be as low as $10^{-4}$ times the standard error of the noise. The normal prior performs better than the unshrunk selection procedure, but neither has regret approaching zero. These patterns are consistent across different noise levels, although the regret are lower with less noise, as the oracle prior and estimated prior are closer.

\begin{figure}[htbp]
    \centering
    \includegraphics[width=\linewidth]{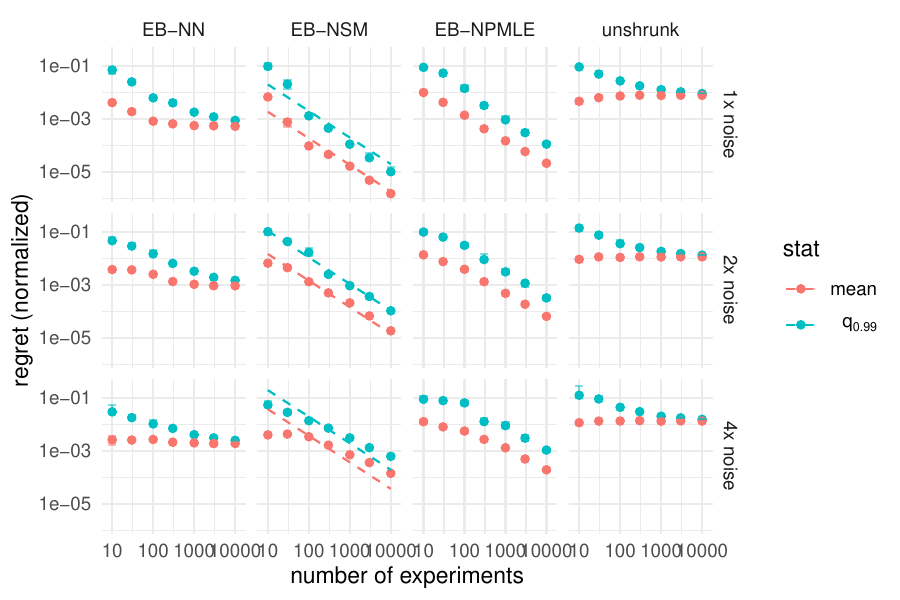}
    \caption{Regret $\regret$ as the number of experiments $n$ increases, on a log scale and normalized by the standard deviation of the noise. The $95\%$ confidence intervals are due to simulation uncertainty. For choice sets based on the correctly specified EB-NSM and EB-NPMLE models, both the $99$th percentile and the mean show a trend of $\bigO(n^{-1})$. Regret does not appear to converge to zero for choice sets based on EB-NN or the unshrunk estimates.}
    \label{fig:regret}
\end{figure}

We compute other quantities of interest from \eqref{eq:two-parts}, such as the proportion of mistakes in \Cref{fig:prop-mistakes} and the maximum magnitude of shrinkage error in \Cref{fig:magnitude}, as well as the $1$-Wasserstein distance between the true prior and the estimated prior in \Cref{fig:w1}. As expected, we see that the proportion of mistakes, their magnitude, and the 1-Wasserstein distance between the true and estimated prior in the correctly specified EB-NSM model all converge to zero at $n^{-1/2}$. The misspecified EB-NN model and the unshrunk procedure perform poorly in comparison, with the proportion and magnitude of mistakes not converging to zero, or even increasing, with the number of experiments. The most flexible model, EB-NPMLE, performs worse along every dimension than the more parsimonious EB-NSM, although the proportion and magnitude of its mistakes both converge to zero. 

\begin{figure}[htbp]
    \centering
    \includegraphics[width=\linewidth]{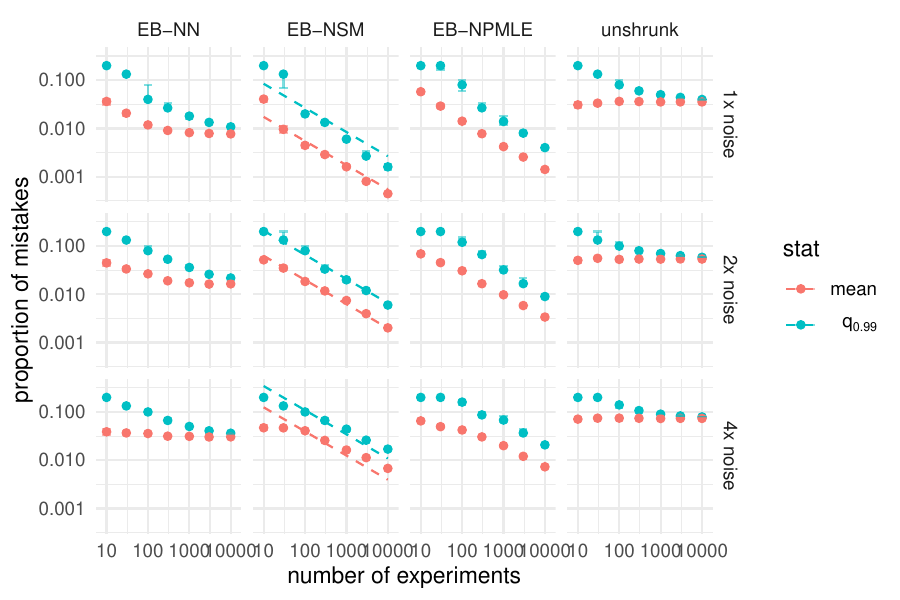}
    \caption{Proportion of mistakes as $n$ increases, on a log scale. The $95\%$ confidence intervals are due to simulation uncertainty. For the correctly specified $\J{EB-NSM}$, both the $99$th percentile and the mean show a trend of $\bigO(n^{-1/2})$. The highly flexible NPMLE shows a similar trend but generally makes more mistakes. The proportion does not appear to decrease towards zero for $\J{EB-NN}$ or $\J{UN}$.}
    \label{fig:prop-mistakes}
\end{figure}

\begin{figure}[htbp]
    \centering
    \includegraphics[width=\linewidth]{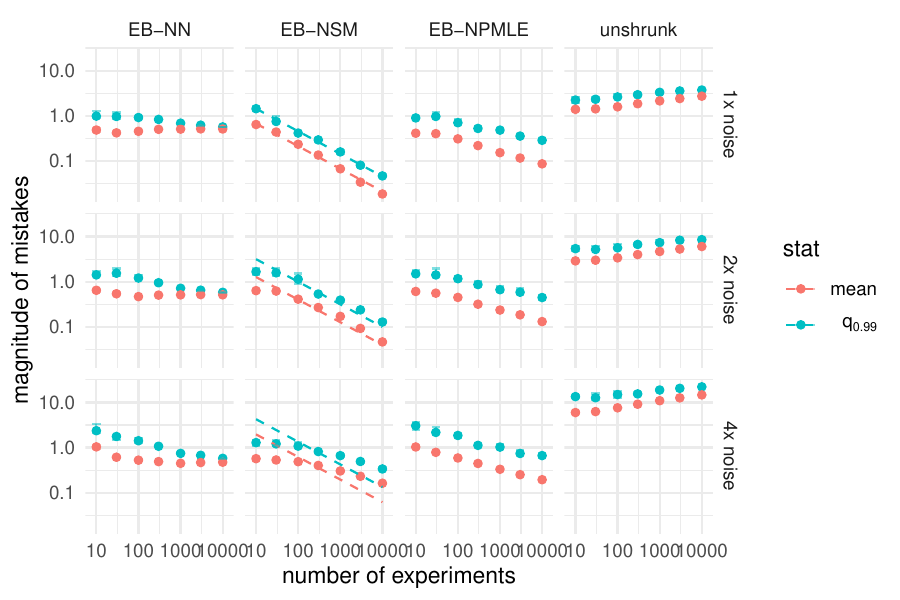}
    \caption{Maximum magnitude of shrinkage error as $n$ increases, on a log scale and normalized by the standard deviation of the noise. The $95\%$ confidence intervals are due to simulation uncertainty. For the correctly specified $\J{EB-NSM}$, both the $99$th percentile and the mean show a trend of $\bigO(n^{-1/2})$ in the lower noise settings. The maximum magnitude does not appear to decrease indefinitely for $\J{EB-NN}$ or $\J{UN}$.}
    \label{fig:magnitude}
\end{figure}

\begin{figure}[htbp]
    \centering
    \includegraphics[width=\linewidth]{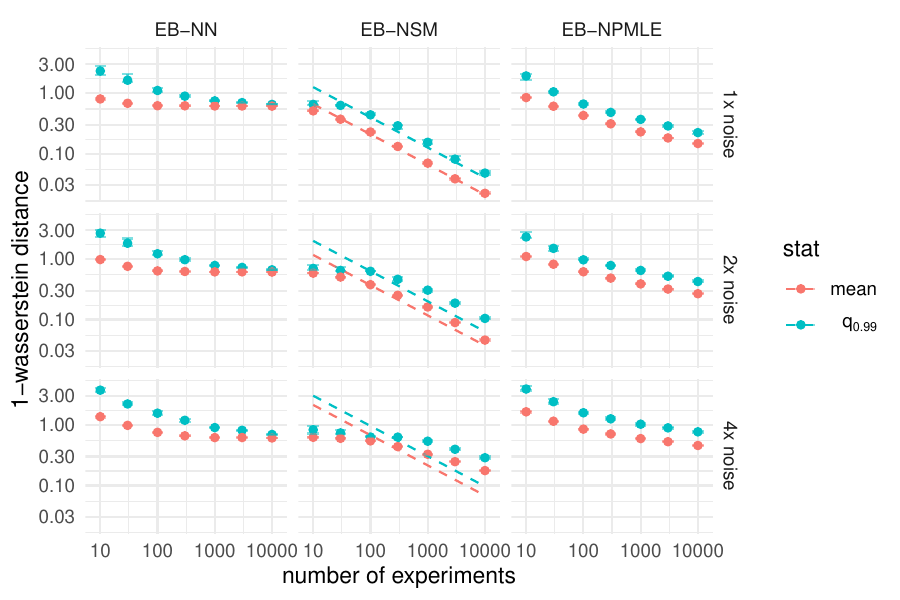}
    \caption{The $1$-Wasserstein distance between the true prior and the estimated prior, as $n$ increases, on a log scale and normalized by the standard deviation of the noise. The $95\%$ confidence intervals are due to simulation uncertainty. As $n$ gets large, both the $99$th percentile and the mean show a trend of $\bigO(n^{-1/2})$ for EB-NSM. The distance levels off away from zero for EB-NN because of misspecification. The distance for EB-NPMLE decreases at a slower rate than $\bigO(n^{-1/2})$.}
    \label{fig:w1}
\end{figure}

\subsection{Estimated standard error}

The simulations above assume the known standard error to be known, which is reasonable for large-scale online experiments where each experiments have million of units. We complement the simulations above to demonstrate how the noise in the estimated standard error will affect the performance of empirical Bayes methods, showing the regret as the number of units increase and the estimation for standard error improves in \Cref{fig:regret-est-se}.

\begin{figure}[htbp]
    \centering
    \includegraphics[width=\linewidth]{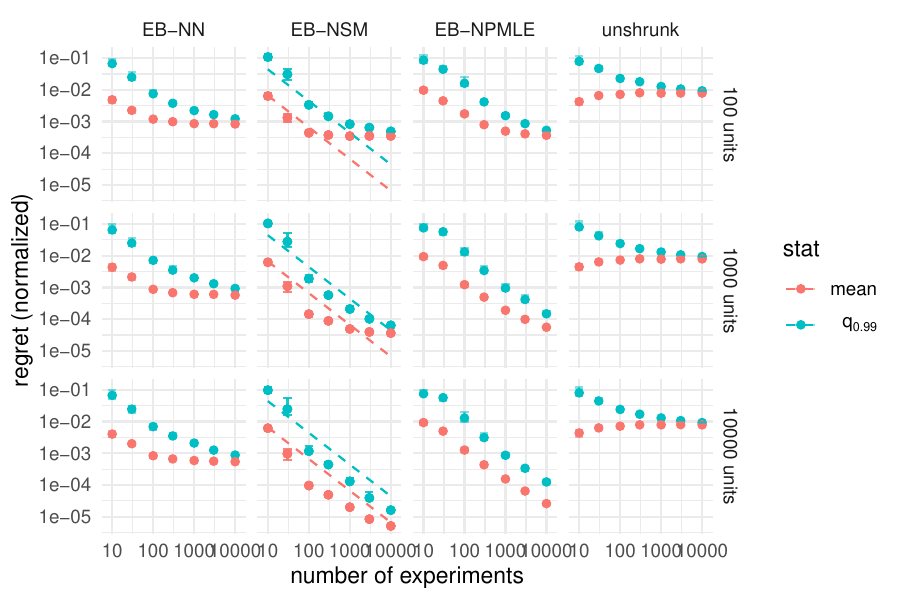}
    \caption{Regret $\regret$ as the number of experiments $n$ increases, on a log scale and normalized by the standard deviation of the noise. Error bars around each point are the $95\%$ confidence intervals from uncertainty due to simulation error. Our result holds better as the number of units increase.}
    \label{fig:regret-est-se}
\end{figure}

\section{Conclusion}

Our results show that empirical Bayes methods perform well in maximizing the aggregate value of the selected units, in the sense that the regret they incur converges to zero faster than the estimation error in the values themselves. This stands in contrast to prior work emphasizing the difficulty of accurately selecting the best units when the decision-maker incurs a discrete loss from each misclassification (e.g.\ \citet{lockwood2002uncertainty,lin2006loss,gu2020invidious}). This underscores that rather than selection being an inherently difficult problem, it depends on whether misclassification errors should be weighted by their severity in the utility function. Finally, we note that many extensions and variations on this setting are yet to be fully explored, including characterizing the performance of decision rules for the frequentist analog of the Bayesian regret we study, treating the true values of units as non-stochastic;\footnote{Specifically, studying the utility $U(\J{EB})$ under permutation invariance as outlined in \citet{weinstein2021permutation}.} improving performance by incorporating unit-specific covariates into the analysis; and extending to an empirical Bayes knapsack problem where the selected units incur heterogeneous costs.

As discussed in \Cref{sec:intro}, the frequentist optimal solution requires unavailable oracle knowledge of the order statistics of $\mu_i$'s. This implies that the empirical Bayes solution is not optimal in a frequentist sense, with mainly two gaps:
\begin{inlinelist}
    \item the optimal solution in \citet[Theorem 1]{weinstein2021permutation} is the Bayesian solution with a uniform prior on the permutations of $\mu_i$'s, while empirical Bayes uses $\hG$ instead;
    \item \citet{weinstein2021permutation} focused on the loss for a specific set of $\mu_i$'s, while our analysis averages this over $G_0$.
\end{inlinelist}

We suspect these gaps are small. For the first gap, \citet[Section 6]{weinstein2021permutation} conjectures that the Bayesian solution using the e.c.d.f.\ of $\mu_i$'s is asymptotically optimally. We believe this can be reasonably recovered as $\hG$ when the class of priors $\mathcal{M}$ is sufficiently large. For the second gap, \citet[Theorem 1]{weinstein2021permutation} showed that minimizing the loss is equivalently to minimizing the loss averaged over a uniform permutation of $\mu_i$'s. Asymptotically this should be close to the loss averaged over $G_0$, our regret $\regret$. Putting this together, both the solution and loss function are similar between the frequentist and the empirical Bayes settings, hinting at some loose frequentist optimality of the empirical Bayes approach.

\bibliographystyle{ACM-Reference-Format}
\bibliography{papers}


\begin{thebibliography}{66}


\ifx \showCODEN    \undefined \def \showCODEN     #1{\unskip}     \fi
\ifx \showISBNx    \undefined \def \showISBNx     #1{\unskip}     \fi
\ifx \showISBNxiii \undefined \def \showISBNxiii  #1{\unskip}     \fi
\ifx \showISSN     \undefined \def \showISSN      #1{\unskip}     \fi
\ifx \showLCCN     \undefined \def \showLCCN      #1{\unskip}     \fi
\ifx \shownote     \undefined \def \shownote      #1{#1}          \fi
\ifx \showarticletitle \undefined \def \showarticletitle #1{#1}   \fi
\ifx \showURL      \undefined \def \showURL       {\relax}        \fi
\providecommand\bibfield[2]{#2}
\providecommand\bibinfo[2]{#2}
\providecommand\natexlab[1]{#1}
\providecommand\showeprint[2][]{arXiv:#2}

\bibitem[Andrews et~al\mbox{.}(2019)]%
        {andrews2019inference}
\bibfield{author}{\bibinfo{person}{Isaiah Andrews}, \bibinfo{person}{Toru
  Kitagawa}, {and} \bibinfo{person}{Adam McCloskey}.}
  \bibinfo{year}{2019}\natexlab{}.
\newblock \bibinfo{booktitle}{\emph{Inference on winners}}.
\newblock \bibinfo{type}{{T}echnical {R}eport}. \bibinfo{institution}{National
  Bureau of Economic Research}.
\newblock


\bibitem[Azevedo et~al\mbox{.}(2020)]%
        {azevedo2020b}
\bibfield{author}{\bibinfo{person}{Eduardo~M Azevedo}, \bibinfo{person}{Alex
  Deng}, \bibinfo{person}{Jos{\'e}~Luis Montiel~Olea}, \bibinfo{person}{Justin
  Rao}, {and} \bibinfo{person}{E~Glen Weyl}.} \bibinfo{year}{2020}\natexlab{}.
\newblock \showarticletitle{A/B testing with fat tails}.
\newblock \bibinfo{journal}{\emph{Journal of Political Economy}}
  \bibinfo{volume}{128}, \bibinfo{number}{12} (\bibinfo{year}{2020}),
  \bibinfo{pages}{4614--000}.
\newblock


\bibitem[Azevedo et~al\mbox{.}(2019)]%
        {azevedo2019empirical}
\bibfield{author}{\bibinfo{person}{Eduardo~M Azevedo}, \bibinfo{person}{Alex
  Deng}, \bibinfo{person}{Jos{\'e}~L Montiel~Olea}, {and}
  \bibinfo{person}{E~Glen Weyl}.} \bibinfo{year}{2019}\natexlab{}.
\newblock \showarticletitle{Empirical bayes estimation of treatment effects
  with many A/B tests: An overview}. In \bibinfo{booktitle}{\emph{AEA Papers
  and Proceedings}}, Vol.~\bibinfo{volume}{109}. \bibinfo{pages}{43--47}.
\newblock


\bibitem[Bergman et~al\mbox{.}(2019)]%
        {bergman2019creating}
\bibfield{author}{\bibinfo{person}{Peter Bergman}, \bibinfo{person}{Raj
  Chetty}, \bibinfo{person}{Stefanie DeLuca}, \bibinfo{person}{Nathaniel
  Hendren}, \bibinfo{person}{Lawrence~F Katz}, {and}
  \bibinfo{person}{Christopher Palmer}.} \bibinfo{year}{2019}\natexlab{}.
\newblock \bibinfo{booktitle}{\emph{Creating moves to opportunity: Experimental
  evidence on barriers to neighborhood choice}}.
\newblock \bibinfo{type}{{T}echnical {R}eport}. \bibinfo{institution}{National
  Bureau of Economic Research}.
\newblock


\bibitem[Brown(1986)]%
        {brown1986fundamentals}
\bibfield{author}{\bibinfo{person}{Lawrence~D. Brown}.}
  \bibinfo{year}{1986}\natexlab{}.
\newblock \bibinfo{booktitle}{\emph{Fundamentals of statistical exponential
  families with applications in statistical decision theory}}.
\newblock \bibinfo{publisher}{Institute of Mathematical Statistics}.
\newblock


\bibitem[Brown(2008)]%
        {brown2008season}
\bibfield{author}{\bibinfo{person}{Lawrence~D Brown}.}
  \bibinfo{year}{2008}\natexlab{}.
\newblock \showarticletitle{In-season prediction of batting averages: A field
  test of empirical Bayes and Bayes methodologies}.
\newblock \bibinfo{journal}{\emph{The Annals of Applied Statistics}}
  \bibinfo{volume}{2}, \bibinfo{number}{1} (\bibinfo{year}{2008}),
  \bibinfo{pages}{113--152}.
\newblock


\bibitem[Carroll and Hall(1988)]%
        {carroll1988optimal}
\bibfield{author}{\bibinfo{person}{Raymond~J Carroll} {and}
  \bibinfo{person}{Peter Hall}.} \bibinfo{year}{1988}\natexlab{}.
\newblock \showarticletitle{Optimal rates of convergence for deconvolving a
  density}.
\newblock \bibinfo{journal}{\emph{J. Amer. Statist. Assoc.}}
  \bibinfo{volume}{83}, \bibinfo{number}{404} (\bibinfo{year}{1988}),
  \bibinfo{pages}{1184--1186}.
\newblock


\bibitem[Chen(1995)]%
        {chen1995optimal}
\bibfield{author}{\bibinfo{person}{Jiahua Chen}.}
  \bibinfo{year}{1995}\natexlab{}.
\newblock \showarticletitle{Optimal rate of convergence for finite mixture
  models}.
\newblock \bibinfo{journal}{\emph{The Annals of Statistics}}
  (\bibinfo{year}{1995}), \bibinfo{pages}{221--233}.
\newblock


\bibitem[Chen(2022)]%
        {chen2022empirical}
\bibfield{author}{\bibinfo{person}{Jiafeng Chen}.}
  \bibinfo{year}{2022}\natexlab{}.
\newblock \showarticletitle{Empirical Bayes when estimation precision predicts
  parameters}.
\newblock \bibinfo{journal}{\emph{arXiv preprint}} (\bibinfo{year}{2022}).
\newblock


\bibitem[Chen et~al\mbox{.}(2017)]%
        {chen2017top}
\bibfield{author}{\bibinfo{person}{Lijie Chen}, \bibinfo{person}{Jian Li},
  {and} \bibinfo{person}{Mingda Qiao}.} \bibinfo{year}{2017}\natexlab{}.
\newblock \showarticletitle{{Nearly Instance Optimal Sample Complexity Bounds
  for Top-k Arm Selection}}. In \bibinfo{booktitle}{\emph{Proceedings of the
  20th International Conference on Artificial Intelligence and Statistics}}
  \emph{(\bibinfo{series}{Proceedings of Machine Learning Research},
  Vol.~\bibinfo{volume}{54})}, \bibfield{editor}{\bibinfo{person}{Aarti Singh}
  {and} \bibinfo{person}{Jerry Zhu}} (Eds.). \bibinfo{publisher}{PMLR},
  \bibinfo{pages}{101--110}.
\newblock


\bibitem[Chetty et~al\mbox{.}(2014a)]%
        {chetty2014measuringI}
\bibfield{author}{\bibinfo{person}{Raj Chetty}, \bibinfo{person}{John~N
  Friedman}, {and} \bibinfo{person}{Jonah~E Rockoff}.}
  \bibinfo{year}{2014}\natexlab{a}.
\newblock \showarticletitle{Measuring the impacts of teachers {I}: Evaluating
  bias in teacher value-added estimates}.
\newblock \bibinfo{journal}{\emph{American economic review}}
  \bibinfo{volume}{104}, \bibinfo{number}{9} (\bibinfo{year}{2014}),
  \bibinfo{pages}{2593--2632}.
\newblock


\bibitem[Chetty et~al\mbox{.}(2014b)]%
        {chetty2014measuringII}
\bibfield{author}{\bibinfo{person}{Raj Chetty}, \bibinfo{person}{John~N
  Friedman}, {and} \bibinfo{person}{Jonah~E Rockoff}.}
  \bibinfo{year}{2014}\natexlab{b}.
\newblock \showarticletitle{Measuring the impacts of teachers {II}: Teacher
  value-added and student outcomes in adulthood}.
\newblock \bibinfo{journal}{\emph{American economic review}}
  \bibinfo{volume}{104}, \bibinfo{number}{9} (\bibinfo{year}{2014}),
  \bibinfo{pages}{2633--79}.
\newblock


\bibitem[Coey and Cunningham(2019)]%
        {coey2019improving}
\bibfield{author}{\bibinfo{person}{Dominic Coey} {and} \bibinfo{person}{Tom
  Cunningham}.} \bibinfo{year}{2019}\natexlab{}.
\newblock \showarticletitle{Improving treatment effect estimators through
  experiment splitting}. In \bibinfo{booktitle}{\emph{The World Wide Web
  Conference}}. \bibinfo{pages}{285--295}.
\newblock


\bibitem[Cohen and Sackrowitz(1989)]%
        {cohen1989two}
\bibfield{author}{\bibinfo{person}{Arthur Cohen} {and}
  \bibinfo{person}{Harold~B. Sackrowitz}.} \bibinfo{year}{1989}\natexlab{}.
\newblock \showarticletitle{Two stage conditionally unbiased estimators of the
  selected mean}.
\newblock \bibinfo{journal}{\emph{Statistics \& Probability Letters}}
  \bibinfo{volume}{8}, \bibinfo{number}{3} (\bibinfo{year}{1989}),
  \bibinfo{pages}{273--278}.
\newblock
\showISSN{0167-7152}
\href{https://doi.org/10.1016/0167-7152(89)90133-8}{doi:\nolinkurl{10.1016/0167-7152(89)90133-8}}


\bibitem[Dahiya(1974)]%
        {dahiya1974estimation}
\bibfield{author}{\bibinfo{person}{Ram~C. Dahiya}.}
  \bibinfo{year}{1974}\natexlab{}.
\newblock \showarticletitle{Estimation of the Mean of the Selected Population}.
\newblock \bibinfo{journal}{\emph{J. Amer. Statist. Assoc.}}
  \bibinfo{volume}{69}, \bibinfo{number}{345} (\bibinfo{year}{1974}),
  \bibinfo{pages}{226--230}.
\newblock
\href{https://doi.org/10.1080/01621459.1974.10480159}{doi:\nolinkurl{10.1080/01621459.1974.10480159}}


\bibitem[Deng(2015)]%
        {deng2015objective}
\bibfield{author}{\bibinfo{person}{Alex Deng}.}
  \bibinfo{year}{2015}\natexlab{}.
\newblock \showarticletitle{Objective bayesian two sample hypothesis testing
  for online controlled experiments}. In \bibinfo{booktitle}{\emph{Proceedings
  of the 24th International Conference on World Wide Web}}.
  \bibinfo{pages}{923--928}.
\newblock


\bibitem[Deng et~al\mbox{.}(2021)]%
        {deng2021post}
\bibfield{author}{\bibinfo{person}{Alex Deng}, \bibinfo{person}{Yicheng Li},
  \bibinfo{person}{Jiannan Lu}, {and} \bibinfo{person}{Vivek Ramamurthy}.}
  \bibinfo{year}{2021}\natexlab{}.
\newblock \showarticletitle{On Post-Selection Inference in A/B Testing}. In
  \bibinfo{booktitle}{\emph{Proceedings of the 27th ACM SIGKDD Conference on
  Knowledge Discovery \& Data Mining}} (Virtual Event, Singapore)
  \emph{(\bibinfo{series}{KDD '21})}. \bibinfo{publisher}{Association for
  Computing Machinery}, \bibinfo{address}{New York, NY, USA},
  \bibinfo{pages}{2743–2752}.
\newblock
\showISBNx{9781450383325}
\href{https://doi.org/10.1145/3447548.3467129}{doi:\nolinkurl{10.1145/3447548.3467129}}


\bibitem[Dimick et~al\mbox{.}(2010)]%
        {dimick2010ranking}
\bibfield{author}{\bibinfo{person}{Justin~B Dimick}, \bibinfo{person}{Douglas~O
  Staiger}, {and} \bibinfo{person}{John~D Birkmeyer}.}
  \bibinfo{year}{2010}\natexlab{}.
\newblock \showarticletitle{Ranking hospitals on surgical mortality: the
  importance of reliability adjustment}.
\newblock \bibinfo{journal}{\emph{Health services research}}
  \bibinfo{volume}{45}, \bibinfo{number}{6p1} (\bibinfo{year}{2010}),
  \bibinfo{pages}{1614--1629}.
\newblock


\bibitem[Durrett(2009)]%
        {durrett2019probability}
\bibfield{author}{\bibinfo{person}{Rick Durrett}.}
  \bibinfo{year}{2009}\natexlab{}.
\newblock \bibinfo{booktitle}{\emph{Probability: Theory and examples}}.
\newblock \bibinfo{publisher}{Cambridge University Press}.
\newblock


\bibitem[Efron(2011)]%
        {efron2011tweedie}
\bibfield{author}{\bibinfo{person}{Bradley Efron}.}
  \bibinfo{year}{2011}\natexlab{}.
\newblock \showarticletitle{Tweedie’s formula and selection bias}.
\newblock \bibinfo{journal}{\emph{J. Amer. Statist. Assoc.}}
  \bibinfo{volume}{106}, \bibinfo{number}{496} (\bibinfo{year}{2011}),
  \bibinfo{pages}{1602--1614}.
\newblock


\bibitem[Efron(2014)]%
        {efron2014two}
\bibfield{author}{\bibinfo{person}{Bradley Efron}.}
  \bibinfo{year}{2014}\natexlab{}.
\newblock \showarticletitle{Two modeling strategies for empirical Bayes
  estimation}.
\newblock \bibinfo{journal}{\emph{Statistical science: a review journal of the
  Institute of Mathematical Statistics}} \bibinfo{volume}{29},
  \bibinfo{number}{2} (\bibinfo{year}{2014}), \bibinfo{pages}{285}.
\newblock


\bibitem[Efron and Morris(1973)]%
        {efron1973stein}
\bibfield{author}{\bibinfo{person}{Bradley Efron} {and} \bibinfo{person}{Carl
  Morris}.} \bibinfo{year}{1973}\natexlab{}.
\newblock \showarticletitle{Stein's estimation rule and its competitors—an
  empirical Bayes approach}.
\newblock \bibinfo{journal}{\emph{J. Amer. Statist. Assoc.}}
  \bibinfo{volume}{68}, \bibinfo{number}{341} (\bibinfo{year}{1973}),
  \bibinfo{pages}{117--130}.
\newblock


\bibitem[Efron and Morris(1977)]%
        {efron1977stein}
\bibfield{author}{\bibinfo{person}{Bradley Efron} {and} \bibinfo{person}{Carl
  Morris}.} \bibinfo{year}{1977}\natexlab{}.
\newblock \showarticletitle{Stein's paradox in statistics}.
\newblock \bibinfo{journal}{\emph{Scientific American}} \bibinfo{volume}{236},
  \bibinfo{number}{5} (\bibinfo{year}{1977}), \bibinfo{pages}{119--127}.
\newblock


\bibitem[Efron and Tibshirani(2002)]%
        {efron2002empirical}
\bibfield{author}{\bibinfo{person}{Bradley Efron} {and} \bibinfo{person}{Robert
  Tibshirani}.} \bibinfo{year}{2002}\natexlab{}.
\newblock \showarticletitle{Empirical Bayes methods and false discovery rates
  for microarrays}.
\newblock \bibinfo{journal}{\emph{Genetic epidemiology}} \bibinfo{volume}{23},
  \bibinfo{number}{1} (\bibinfo{year}{2002}), \bibinfo{pages}{70--86}.
\newblock


\bibitem[Fan(1991)]%
        {fan1991optimal}
\bibfield{author}{\bibinfo{person}{Jianqing Fan}.}
  \bibinfo{year}{1991}\natexlab{}.
\newblock \showarticletitle{On the optimal rates of convergence for
  nonparametric deconvolution problems}.
\newblock \bibinfo{journal}{\emph{The Annals of Statistics}}
  (\bibinfo{year}{1991}), \bibinfo{pages}{1257--1272}.
\newblock


\bibitem[Fithian et~al\mbox{.}(2014)]%
        {fithian2014optimal}
\bibfield{author}{\bibinfo{person}{William Fithian}, \bibinfo{person}{Dennis
  Sun}, {and} \bibinfo{person}{Jonathan Taylor}.}
  \bibinfo{year}{2014}\natexlab{}.
\newblock \showarticletitle{Optimal inference after model selection}.
\newblock \bibinfo{journal}{\emph{arXiv preprint}} (\bibinfo{year}{2014}).
\newblock


\bibitem[Gilraine et~al\mbox{.}(2020)]%
        {gilraine2020new}
\bibfield{author}{\bibinfo{person}{Michael Gilraine}, \bibinfo{person}{Jiaying
  Gu}, {and} \bibinfo{person}{Robert McMillan}.}
  \bibinfo{year}{2020}\natexlab{}.
\newblock \bibinfo{booktitle}{\emph{A new method for estimating teacher
  value-added}}.
\newblock \bibinfo{type}{{T}echnical {R}eport}. \bibinfo{institution}{National
  Bureau of Economic Research}.
\newblock


\bibitem[Goldberg and Johndrow(2017)]%
        {goldberg2017decision}
\bibfield{author}{\bibinfo{person}{David Goldberg} {and}
  \bibinfo{person}{James~E Johndrow}.} \bibinfo{year}{2017}\natexlab{}.
\newblock \showarticletitle{A decision theoretic approach to A/B testing}.
\newblock \bibinfo{journal}{\emph{arXiv preprint}} (\bibinfo{year}{2017}).
\newblock


\bibitem[Goldstein and Spiegelhalter(1996)]%
        {goldstein1996league}
\bibfield{author}{\bibinfo{person}{Harvey Goldstein} {and}
  \bibinfo{person}{David~J Spiegelhalter}.} \bibinfo{year}{1996}\natexlab{}.
\newblock \showarticletitle{League tables and their limitations: statistical
  issues in comparisons of institutional performance}.
\newblock \bibinfo{journal}{\emph{Journal of the Royal Statistical Society:
  Series A (Statistics in Society)}} \bibinfo{volume}{159}, \bibinfo{number}{3}
  (\bibinfo{year}{1996}), \bibinfo{pages}{385--409}.
\newblock


\bibitem[Gu and Koenker(2020)]%
        {gu2020invidious}
\bibfield{author}{\bibinfo{person}{Jiaying Gu} {and} \bibinfo{person}{Roger
  Koenker}.} \bibinfo{year}{2020}\natexlab{}.
\newblock \showarticletitle{Invidious comparisons: Ranking and selection as
  compound decisions}.
\newblock \bibinfo{journal}{\emph{arXiv preprint}} (\bibinfo{year}{2020}).
\newblock


\bibitem[Guo et~al\mbox{.}(2020)]%
        {guo2020empirical}
\bibfield{author}{\bibinfo{person}{F~Richard Guo}, \bibinfo{person}{James
  McQueen}, {and} \bibinfo{person}{Thomas~S Richardson}.}
  \bibinfo{year}{2020}\natexlab{}.
\newblock \showarticletitle{Empirical Bayes for Large-scale Randomized
  Experiments: a Spectral Approach}.
\newblock \bibinfo{journal}{\emph{arXiv preprint}} (\bibinfo{year}{2020}).
\newblock


\bibitem[Guo and He(2021)]%
        {guo2021inference}
\bibfield{author}{\bibinfo{person}{Xinzhou Guo} {and} \bibinfo{person}{Xuming
  He}.} \bibinfo{year}{2021}\natexlab{}.
\newblock \showarticletitle{Inference on selected subgroups in clinical
  trials}.
\newblock \bibinfo{journal}{\emph{J. Amer. Statist. Assoc.}}
  \bibinfo{volume}{116}, \bibinfo{number}{535} (\bibinfo{year}{2021}),
  \bibinfo{pages}{1498--1506}.
\newblock


\bibitem[Gupta and Li(2005)]%
        {gupta2005empirical}
\bibfield{author}{\bibinfo{person}{Shanti~S Gupta} {and}
  \bibinfo{person}{Jianjun Li}.} \bibinfo{year}{2005}\natexlab{}.
\newblock \showarticletitle{On empirical Bayes procedures for selecting good
  populations in a positive exponential family}.
\newblock \bibinfo{journal}{\emph{Journal of Statistical planning and
  Inference}} \bibinfo{volume}{129}, \bibinfo{number}{1-2}
  (\bibinfo{year}{2005}), \bibinfo{pages}{3--18}.
\newblock


\bibitem[Gupta and Panchapakesan(2002)]%
        {gupta2002multiple}
\bibfield{author}{\bibinfo{person}{Shanti~S Gupta} {and}
  \bibinfo{person}{Subramanian Panchapakesan}.}
  \bibinfo{year}{2002}\natexlab{}.
\newblock \bibinfo{booktitle}{\emph{Multiple decision procedures: theory and
  methodology of selecting and ranking populations}}.
\newblock \bibinfo{publisher}{SIAM}.
\newblock


\bibitem[Hannan and Van~Ryzin(1965)]%
        {hannan1965rate}
\bibfield{author}{\bibinfo{person}{James~F Hannan} {and} \bibinfo{person}{JR
  Van~Ryzin}.} \bibinfo{year}{1965}\natexlab{}.
\newblock \showarticletitle{Rate of convergence in the compound decision
  problem for two completely specified distributions}.
\newblock \bibinfo{journal}{\emph{The Annals of Mathematical Statistics}}
  (\bibinfo{year}{1965}), \bibinfo{pages}{1743--1752}.
\newblock


\bibitem[Harris and Sass(2014)]%
        {harris2014skills}
\bibfield{author}{\bibinfo{person}{Douglas~N Harris} {and}
  \bibinfo{person}{Tim~R Sass}.} \bibinfo{year}{2014}\natexlab{}.
\newblock \showarticletitle{Skills, productivity and the evaluation of teacher
  performance}.
\newblock \bibinfo{journal}{\emph{Economics of Education Review}}
  \bibinfo{volume}{40} (\bibinfo{year}{2014}), \bibinfo{pages}{183--204}.
\newblock


\bibitem[Hull(2018)]%
        {hull2018estimating}
\bibfield{author}{\bibinfo{person}{Peter Hull}.}
  \bibinfo{year}{2018}\natexlab{}.
\newblock \showarticletitle{Estimating hospital quality with quasi-experimental
  data}.
\newblock \bibinfo{journal}{\emph{Available at SSRN 3118358}}
  (\bibinfo{year}{2018}).
\newblock


\bibitem[Hung and Fithian(2019)]%
        {hung2019rank}
\bibfield{author}{\bibinfo{person}{Kenneth Hung} {and} \bibinfo{person}{William
  Fithian}.} \bibinfo{year}{2019}\natexlab{}.
\newblock \showarticletitle{{Rank verification for exponential families}}.
\newblock \bibinfo{journal}{\emph{The Annals of Statistics}}
  \bibinfo{volume}{47}, \bibinfo{number}{2} (\bibinfo{year}{2019}),
  \bibinfo{pages}{758 -- 782}.
\newblock
\href{https://doi.org/10.1214/17-AOS1634}{doi:\nolinkurl{10.1214/17-AOS1634}}


\bibitem[Jacob and Lefgren(2008)]%
        {jacob2008can}
\bibfield{author}{\bibinfo{person}{Brian~A Jacob} {and} \bibinfo{person}{Lars
  Lefgren}.} \bibinfo{year}{2008}\natexlab{}.
\newblock \showarticletitle{Can principals identify effective teachers?
  Evidence on subjective performance evaluation in education}.
\newblock \bibinfo{journal}{\emph{Journal of labor Economics}}
  \bibinfo{volume}{26}, \bibinfo{number}{1} (\bibinfo{year}{2008}),
  \bibinfo{pages}{101--136}.
\newblock


\bibitem[Kane et~al\mbox{.}(2008)]%
        {kane2008does}
\bibfield{author}{\bibinfo{person}{Thomas~J Kane}, \bibinfo{person}{Jonah~E
  Rockoff}, {and} \bibinfo{person}{Douglas~O Staiger}.}
  \bibinfo{year}{2008}\natexlab{}.
\newblock \showarticletitle{What does certification tell us about teacher
  effectiveness? Evidence from New York City}.
\newblock \bibinfo{journal}{\emph{Economics of Education review}}
  \bibinfo{volume}{27}, \bibinfo{number}{6} (\bibinfo{year}{2008}),
  \bibinfo{pages}{615--631}.
\newblock


\bibitem[Keener(2010)]%
        {keener2010theoretical}
\bibfield{author}{\bibinfo{person}{Robert~W Keener}.}
  \bibinfo{year}{2010}\natexlab{}.
\newblock \bibinfo{booktitle}{\emph{Theoretical statistics: Topics for a core
  course}}.
\newblock \bibinfo{publisher}{Springer}.
\newblock


\bibitem[Kiefer and Wolfowitz(1956)]%
        {kiefer1956consistency}
\bibfield{author}{\bibinfo{person}{Jack Kiefer} {and} \bibinfo{person}{Jacob
  Wolfowitz}.} \bibinfo{year}{1956}\natexlab{}.
\newblock \showarticletitle{Consistency of the maximum likelihood estimator in
  the presence of infinitely many incidental parameters}.
\newblock \bibinfo{journal}{\emph{The Annals of Mathematical Statistics}}
  (\bibinfo{year}{1956}), \bibinfo{pages}{887--906}.
\newblock


\bibitem[Kline and Walters(2021)]%
        {kline2021reasonable}
\bibfield{author}{\bibinfo{person}{Patrick Kline} {and}
  \bibinfo{person}{Christopher Walters}.} \bibinfo{year}{2021}\natexlab{}.
\newblock \showarticletitle{Reasonable Doubt: Experimental Detection of
  Job-Level Employment Discrimination}.
\newblock \bibinfo{journal}{\emph{Econometrica}} \bibinfo{volume}{89},
  \bibinfo{number}{2} (\bibinfo{year}{2021}), \bibinfo{pages}{765--792}.
\newblock


\bibitem[Koenker and Mizera(2014)]%
        {koenker2014convex}
\bibfield{author}{\bibinfo{person}{Roger Koenker} {and} \bibinfo{person}{Ivan
  Mizera}.} \bibinfo{year}{2014}\natexlab{}.
\newblock \showarticletitle{Convex optimization, shape constraints, compound
  decisions, and empirical Bayes rules}.
\newblock \bibinfo{journal}{\emph{J. Amer. Statist. Assoc.}}
  \bibinfo{volume}{109}, \bibinfo{number}{506} (\bibinfo{year}{2014}),
  \bibinfo{pages}{674--685}.
\newblock


\bibitem[Lin et~al\mbox{.}(2006)]%
        {lin2006loss}
\bibfield{author}{\bibinfo{person}{Rongheng Lin}, \bibinfo{person}{Thomas~A
  Louis}, \bibinfo{person}{Susan~M Paddock}, {and} \bibinfo{person}{Greg
  Ridgeway}.} \bibinfo{year}{2006}\natexlab{}.
\newblock \showarticletitle{Loss function based ranking in two-stage,
  hierarchical models}.
\newblock \bibinfo{journal}{\emph{Bayesian Analysis}} \bibinfo{volume}{1},
  \bibinfo{number}{4} (\bibinfo{year}{2006}), \bibinfo{pages}{915}.
\newblock


\bibitem[Lockwood et~al\mbox{.}(2002)]%
        {lockwood2002uncertainty}
\bibfield{author}{\bibinfo{person}{JR Lockwood}, \bibinfo{person}{Thomas~A
  Louis}, {and} \bibinfo{person}{Daniel~F McCaffrey}.}
  \bibinfo{year}{2002}\natexlab{}.
\newblock \showarticletitle{Uncertainty in rank estimation: Implications for
  value-added modeling accountability systems}.
\newblock \bibinfo{journal}{\emph{Journal of educational and behavioral
  statistics}} \bibinfo{volume}{27}, \bibinfo{number}{3}
  (\bibinfo{year}{2002}), \bibinfo{pages}{255--270}.
\newblock


\bibitem[Marshall(1991)]%
        {marshall1991mapping}
\bibfield{author}{\bibinfo{person}{Roger~J Marshall}.}
  \bibinfo{year}{1991}\natexlab{}.
\newblock \showarticletitle{Mapping disease and mortality rates using empirical
  Bayes estimators}.
\newblock \bibinfo{journal}{\emph{Journal of the Royal Statistical Society:
  Series C (Applied Statistics)}} \bibinfo{volume}{40}, \bibinfo{number}{2}
  (\bibinfo{year}{1991}), \bibinfo{pages}{283--294}.
\newblock


\bibitem[Matias et~al\mbox{.}(2021)]%
        {matias2021upworthy}
\bibfield{author}{\bibinfo{person}{J.~Nathan Matias}, \bibinfo{person}{Kevin
  Munger}, \bibinfo{person}{Marianne~Aubin Le~Quere}, {and}
  \bibinfo{person}{Charles Ebersole}.} \bibinfo{year}{2021}\natexlab{}.
\newblock \showarticletitle{The Upworthy Research Archive, a time series of
  32,487 experiments in U.S. media}.
\newblock \bibinfo{journal}{\emph{Scientific Data}} \bibinfo{volume}{8},
  \bibinfo{number}{1} (\bibinfo{year}{2021}), \bibinfo{pages}{195}.
\newblock
\showISBNx{2052-4463}
\href{https://doi.org/10.1038/s41597-021-00934-7}{doi:\nolinkurl{10.1038/s41597-021-00934-7}}


\bibitem[Mogstad et~al\mbox{.}(2022)]%
        {mogstad2022comment}
\bibfield{author}{\bibinfo{person}{M Mogstad}, \bibinfo{person}{J Romano},
  \bibinfo{person}{A Shaikh}, {and} \bibinfo{person}{D Wilhelm}.}
  \bibinfo{year}{2022}\natexlab{}.
\newblock \showarticletitle{Comment on" Invidious Comparisons: Ranking and
  Selection as Compound Decisions"}.
\newblock \bibinfo{journal}{\emph{Econometrica}} (\bibinfo{year}{2022}).
\newblock


\bibitem[Mogstad et~al\mbox{.}(2024)]%
        {mogstad2024inference}
\bibfield{author}{\bibinfo{person}{Magne Mogstad}, \bibinfo{person}{Joseph~P
  Romano}, \bibinfo{person}{Azeem~M Shaikh}, {and} \bibinfo{person}{Daniel
  Wilhelm}.} \bibinfo{year}{2024}\natexlab{}.
\newblock \showarticletitle{Inference for ranks with applications to mobility
  across neighbourhoods and academic achievement across countries}.
\newblock \bibinfo{journal}{\emph{Review of Economic Studies}}
  \bibinfo{volume}{91}, \bibinfo{number}{1} (\bibinfo{year}{2024}),
  \bibinfo{pages}{476--518}.
\newblock


\bibitem[Polyanskiy and Wu(2021)]%
        {polyanskiy2021sharp}
\bibfield{author}{\bibinfo{person}{Yury Polyanskiy} {and}
  \bibinfo{person}{Yihong Wu}.} \bibinfo{year}{2021}\natexlab{}.
\newblock \showarticletitle{Sharp regret bounds for empirical Bayes and
  compound decision problems}.
\newblock \bibinfo{journal}{\emph{arXiv preprint}} (\bibinfo{year}{2021}).
\newblock


\bibitem[Robbins(1951)]%
        {robbins1951asymptotically}
\bibfield{author}{\bibinfo{person}{Herbert Robbins}.}
  \bibinfo{year}{1951}\natexlab{}.
\newblock \showarticletitle{Asymptotically Subminimax Solutions of Compound
  Statistical Decision Problems}. In \bibinfo{booktitle}{\emph{Proceedings of
  the Second Berkeley Symposium on Mathematical Statistics and Probability:
  Volume 2}}. University of California Press, \bibinfo{pages}{131--149}.
\newblock


\bibitem[Robbins(1956)]%
        {robbins1956empirical}
\bibfield{author}{\bibinfo{person}{Herbert Robbins}.}
  \bibinfo{year}{1956}\natexlab{}.
\newblock \showarticletitle{An Empirical Bayes Approach to Statistics}. In
  \bibinfo{booktitle}{\emph{Proceedings of the Third Berkeley Symposium on
  Mathematical Statistics and Probability, Volume 1: Contributions to the
  Theory of Statistics}}. University of California Press,
  \bibinfo{pages}{157--163}.
\newblock


\bibitem[Robbins(1964)]%
        {robbins1964empirical}
\bibfield{author}{\bibinfo{person}{Herbert Robbins}.}
  \bibinfo{year}{1964}\natexlab{}.
\newblock \showarticletitle{The empirical Bayes approach to statistical
  decision problems}.
\newblock \bibinfo{journal}{\emph{The Annals of Mathematical Statistics}}
  \bibinfo{volume}{35}, \bibinfo{number}{1} (\bibinfo{year}{1964}),
  \bibinfo{pages}{1--20}.
\newblock


\bibitem[Shang et~al\mbox{.}(2020)]%
        {shang2020bai}
\bibfield{author}{\bibinfo{person}{Xuedong Shang}, \bibinfo{person}{Rianne de
  Heide}, \bibinfo{person}{Pierre Menard}, \bibinfo{person}{Emilie Kaufmann},
  {and} \bibinfo{person}{Michal Valko}.} \bibinfo{year}{2020}\natexlab{}.
\newblock \showarticletitle{Fixed-confidence guarantees for Bayesian best-arm
  identification}. In \bibinfo{booktitle}{\emph{Proceedings of the Twenty Third
  International Conference on Artificial Intelligence and Statistics}}
  \emph{(\bibinfo{series}{Proceedings of Machine Learning Research},
  Vol.~\bibinfo{volume}{108})}, \bibfield{editor}{\bibinfo{person}{Silvia
  Chiappa} {and} \bibinfo{person}{Roberto Calandra}} (Eds.).
  \bibinfo{publisher}{PMLR}, \bibinfo{pages}{1823--1832}.
\newblock


\bibitem[Smyth(2004)]%
        {smyth2004linear}
\bibfield{author}{\bibinfo{person}{Gordon~K Smyth}.}
  \bibinfo{year}{2004}\natexlab{}.
\newblock \showarticletitle{Linear models and empirical bayes methods for
  assessing differential expression in microarray experiments}.
\newblock \bibinfo{journal}{\emph{Statistical applications in genetics and
  molecular biology}} \bibinfo{volume}{3}, \bibinfo{number}{1}
  (\bibinfo{year}{2004}).
\newblock


\bibitem[Soloff et~al\mbox{.}(2024)]%
        {soloff2024multivariate}
\bibfield{author}{\bibinfo{person}{Jake~A Soloff}, \bibinfo{person}{Adityanand
  Guntuboyina}, {and} \bibinfo{person}{Bodhisattva Sen}.}
  \bibinfo{year}{2024}\natexlab{}.
\newblock \showarticletitle{{Multivariate, heteroscedastic empirical Bayes via
  nonparametric maximum likelihood}}.
\newblock \bibinfo{journal}{\emph{Journal of the Royal Statistical Society
  Series B: Statistical Methodology}} (\bibinfo{date}{05}
  \bibinfo{year}{2024}), \bibinfo{pages}{qkae040}.
\newblock
\showISSN{1369-7412}
\href{https://doi.org/10.1093/jrsssb/qkae040}{doi:\nolinkurl{10.1093/jrsssb/qkae040}}


\bibitem[Stephens(2016)]%
        {stephens2016fdr}
\bibfield{author}{\bibinfo{person}{Matthew Stephens}.}
  \bibinfo{year}{2016}\natexlab{}.
\newblock \showarticletitle{{False discovery rates: a new deal}}.
\newblock \bibinfo{journal}{\emph{Biostatistics}} \bibinfo{volume}{18},
  \bibinfo{number}{2} (\bibinfo{date}{10} \bibinfo{year}{2016}),
  \bibinfo{pages}{275--294}.
\newblock
\showISSN{1465-4644}
\href{https://doi.org/10.1093/biostatistics/kxw041}{doi:\nolinkurl{10.1093/biostatistics/kxw041}}


\bibitem[Thomas et~al\mbox{.}(1994)]%
        {thomas1994empirical}
\bibfield{author}{\bibinfo{person}{Neal Thomas}, \bibinfo{person}{Nicholas~T
  Longford}, {and} \bibinfo{person}{John~E Rolph}.}
  \bibinfo{year}{1994}\natexlab{}.
\newblock \showarticletitle{Empirical Bayes methods for estimating
  hospital-specific mortality rates}.
\newblock \bibinfo{journal}{\emph{Statistics in medicine}}
  \bibinfo{volume}{13}, \bibinfo{number}{9} (\bibinfo{year}{1994}),
  \bibinfo{pages}{889--903}.
\newblock


\bibitem[van~der Vaart(2000)]%
        {van2000asymptotic}
\bibfield{author}{\bibinfo{person}{Aad~W van~der Vaart}.}
  \bibinfo{year}{2000}\natexlab{}.
\newblock \bibinfo{booktitle}{\emph{Asymptotic statistics}}.
  Vol.~\bibinfo{volume}{3}.
\newblock \bibinfo{publisher}{Cambridge university press}.
\newblock


\bibitem[Van~Ryzin and Susarla(1977)]%
        {van1977empirical}
\bibfield{author}{\bibinfo{person}{John Van~Ryzin} {and}
  \bibinfo{person}{Vyaghreswarudu Susarla}.} \bibinfo{year}{1977}\natexlab{}.
\newblock \showarticletitle{On the empirical Bayes approach to multiple
  decision problems}.
\newblock \bibinfo{journal}{\emph{The Annals of Statistics}}
  \bibinfo{volume}{5}, \bibinfo{number}{1} (\bibinfo{year}{1977}),
  \bibinfo{pages}{172--181}.
\newblock


\bibitem[Weinstein(2021)]%
        {weinstein2021permutation}
\bibfield{author}{\bibinfo{person}{Asaf Weinstein}.}
  \bibinfo{year}{2021}\natexlab{}.
\newblock \showarticletitle{On Permutation Invariant Problems in Large-Scale
  Inference}.
\newblock \bibinfo{journal}{\emph{arXiv preprint}} (\bibinfo{year}{2021}).
\newblock


\bibitem[Weinstein et~al\mbox{.}(2018)]%
        {weinstein2018group}
\bibfield{author}{\bibinfo{person}{Asaf Weinstein}, \bibinfo{person}{Zhuang
  Ma}, \bibinfo{person}{Lawrence~D. Brown}, {and} \bibinfo{person}{Cun-Hui
  Zhang}.} \bibinfo{year}{2018}\natexlab{}.
\newblock \showarticletitle{Group-Linear Empirical Bayes Estimates for a
  Heteroscedastic Normal Mean}.
\newblock \bibinfo{journal}{\emph{J. Amer. Statist. Assoc.}}
  \bibinfo{volume}{113}, \bibinfo{number}{522} (\bibinfo{year}{2018}),
  \bibinfo{pages}{698--710}.
\newblock
\href{https://doi.org/10.1080/01621459.2017.1280406}{doi:\nolinkurl{10.1080/01621459.2017.1280406}}


\bibitem[Willwerscheid and Stephens(2021)]%
        {willwerscheid2021ebnm}
\bibfield{author}{\bibinfo{person}{Jason Willwerscheid} {and}
  \bibinfo{person}{Matthew Stephens}.} \bibinfo{year}{2021}\natexlab{}.
\newblock \showarticletitle{ebnm: An {R} Package for Solving the Empirical
  Bayes Normal Means Problem Using a Variety of Prior Families}.
\newblock \bibinfo{journal}{\emph{arXiv preprint}} (\bibinfo{year}{2021}).
\newblock


\bibitem[Yu et~al\mbox{.}(2020)]%
        {yu2020bayes}
\bibfield{author}{\bibinfo{person}{Peng Yu}, \bibinfo{person}{Spencer~S
  Ericksen}, \bibinfo{person}{Anthony Gitter}, {and} \bibinfo{person}{Michael~A
  Newton}.} \bibinfo{year}{2020}\natexlab{}.
\newblock \showarticletitle{Bayes Optimal Informer Sets for Early-Stage Drug
  Discovery}.
\newblock \bibinfo{journal}{\emph{arXiv preprint}} (\bibinfo{year}{2020}).
\newblock


\bibitem[Zhang(1997)]%
        {zhang1997empirical}
\bibfield{author}{\bibinfo{person}{Cun-Hui Zhang}.}
  \bibinfo{year}{1997}\natexlab{}.
\newblock \showarticletitle{Empirical Bayes and compound estimation of normal
  means}.
\newblock \bibinfo{journal}{\emph{Statistica Sinica}} \bibinfo{volume}{7},
  \bibinfo{number}{1} (\bibinfo{year}{1997}), \bibinfo{pages}{181--193}.
\newblock


\end{thebibliography}

\begin{acks}
We thank Eytan Bakshy, Kevin Chen, Matt Goldman, Daniel Jiang, Jelena Markovic, Sepehr Akhavan Masouleh, Jingang Miao, Adam Obeng, Alex Peysakovich, Okke Schrijvers, Yanyi Song, Daniel Ting and Mark Tygert for their comments and suggestions.
\end{acks}

\appendix

\section{Supporting Proofs}
\label{sec:real-analysis}

\begin{lemma}
\label{lma:exp-family}
    For any exponential family $G(\eta)$ with finite variance, the mapping $\eta \mapsto G(\eta)$ is locally Lipschitz with respect to the $L_1$-norm (or equivalently any $L_p$-norm) in the domain and the $1$-Wasserstein distance in the codomain.
\end{lemma}
\begin{proof}
    Suppose the exponential family is given by
    \[
        \exp(\eta' T(x) - A(\eta)) h(x) \,d\mu(x)
    \]
    with $X$ as the variable and $T = T(X)$ as the canonical statistic. We wish to show that $| \int f(x) \,dG(\eta_1) - \int f(x) \,dG(\eta_0) | / \| \eta_1 - \eta_0 \|_1$ bounded for $\eta_0 \ne \eta_1$, locally in $\eta$ and uniformly over all functions $f$ with Lipschitz constant $1$. By mean value theorem,
    \begin{align*}
        &\left|\int f(x) \,dG(\eta_1) - \int f(x) \,dG(\eta_0)\right| \\
        ={}& \left|(\eta_1 - \eta_0)' \left.\nabla_\eta \int f(x) \,dG(\eta)\right|_{\eta = \eta'}\right| \\
        \le{}& \| \eta_1 - \eta_0 \|_1 \left\| \left.\nabla_\eta \int f(x) \,dG(\eta)\right|_{\eta = \eta'} \right\|_\infty.
    \end{align*}
    for some $\eta'$ that is also local. By \citet[Theorem 2.4]{keener2010theoretical},
    \begin{align*}
        \nabla_\eta \int f(x) \,dG(\eta) &= \int f(x) (T(x) - A'(\eta)) \,dG(\eta) \\
        &= \cov_\eta(f(X), T(X)).
    \end{align*}
    For $i$-th component of the covariance vector, we have
    \begin{align*}
        | \cov_\eta(f(X), T(X))_i | &= | \cov_\eta(f(X), T_i(X)) | \\
        &\le \sqrt{\var_\eta f(X) \var_\eta T_i(X)}.
    \end{align*}
    $\var_\eta T_i(X)$ is given by $A''(\eta)_{ii}$, which is continuous by \citet[Theorem 2.2]{brown1986fundamentals} and thus locally bounded in $\eta$. For $\var_\eta f(X)$, suppose $X'$ is an i.i.d.\ copy of $X$, then
    \begin{align*}
        \var_\eta f(X) & = \frac{1}{2} (\var_\eta f(X) + \var_\eta f(X')) \\
        & = \frac{1}{2} \var_\eta [f(X) - \EE_\eta f(X) + \EE_\eta f(X') - f(X')] \\
        & = \var_\eta [f(X) - f(X')] \\
        & \le \var_\eta (X - X') \\
        & \le \var_\eta X,
    \end{align*}
    which is also continuous and thus locally bounded in $\eta$.
\end{proof}

\begin{lemma}
\label{lma:shrinkage-increasing}
    Under \Cref{asm:compact-h}, the posterior mean function $f_{G, \sigma}$ is strictly increasing and differentiable in $X$, and thus admits an inverse $f_{G, \sigma}^{-1}$ over its image.
\end{lemma}
\begin{proof}
    From \citet{efron2011tweedie}, under $\mu \sim G$ and $X \mid \mu \sim \Normal(\mu, \sigma^2)$, we have $ \nabla_x f_{G, \sigma}(x) = \sigma^{-2} \var(\mu \mid X = x) > 0$.
\end{proof}

\begin{lemma}
\label{lma:p-diff}
    With \Cref{asm:compact-h}, the c.d.f.\ $P$ of $\theta_i$ is continuously differentiable with positive derivative, or equivalently, $\theta_i$ has positive continuous density.
\end{lemma}
\begin{proof}
    The characteristic function of $X \mid \sigma$ is given by $\varphi_G(t) \allowbreak \exp(-\sigma^2 t^2 / 2)$, where $\varphi_G$ is the characteristic function of $G$. Since $|\varphi_G(t) \exp(-\sigma^2 t^2 / 2)|$ is bounded by $\exp(-\sigma^2 t^2 / 2)$ which is integrable, $X \mid \sigma$ has bounded continuous density (\citet[Theorem 3.3.14]{durrett2019probability}). In fact the density is given by
    \[
        p(X \mid \sigma) = \frac{1}{2\pi} \int e^{-itX} \varphi_G(t) \exp(-\sigma^2 t^2 / 2) \,dt.
    \]
    Since $\sigma$ is bounded away from $0$, dominated convergence theorem implies joint continuity of the density above in $(\sigma, X)$. In fact, by dominated convergence theorem and the fact that $t^k \exp(-\sigma^2 t^2 / 2)$ is integrable for all integer $k \ge 0$, we can see that all higher derivatives of the density with respect to $X$ are continuous in $(\sigma, X)$.
    
    Consider the mapping $X \mapsto \theta = f_{G_0, \sigma}(X)$. By \Cref{lma:shrinkage-increasing} it has a strictly positive derivative. Furthermore, since the derivative can be written in terms of the derivatives of $p(X \mid \sigma)$ (\citet{efron2011tweedie}), it is also continuous in $(\sigma, X)$. In other words, the density $p(\theta \mid \sigma)$ is continuous in $(\sigma, X)$.
    
    \Cref{asm:compact-h} assumes the support of $\sigma$ is compact, so the density $p(\theta \mid \sigma)$ is naturally pointwise equicontinuous when viewed as a family of functions indexed by $\sigma$. Now for any $\theta$ and any $\varepsilon > 0$, we can select $\delta > 0$ such that for all $\sigma$ and all $\theta'$ with $|\theta' - \theta| < \delta$, we have $|p(\theta' \mid \sigma) - p(\theta \mid \sigma)| < \varepsilon$ and so
    \begin{align*}
        & \left| \int p(\theta' \mid \sigma) \,dH_0 - \int p(\theta \mid \sigma) \,dH_0 \right| \\
        \le{}& \int |p(\theta' \mid \sigma) - p(\theta \mid \sigma)| \,dH_0 < \varepsilon,
    \end{align*}
    and the marginal density of $\theta$ is continuous.
\end{proof}

\begin{lemma}
\label{lma:stack-ex}
    Let $A$ be a metric space. Suppose $f(a, x)$ as a function from $A \times \RR$ to $\RR$ is continuous and has an inverse with respect to $x$, i.e.\ for all $a$ there exists $f^{-1}_a(\cdot)$ such that $f^{-1}_a \circ f(a, \cdot) = \text{id}_\RR$. Then $(a, y) \mapsto f_a^{-1}(y)$ is also continuous.
\end{lemma}
\begin{proof}
    Let $(a_n, y_n) \to (a^*, y^*)$. It suffices to show that
    \[
        f_{a_n}^{-1}(y_n) \to f_{a^*}^{-1}(y^*).
    \]

    We first show that the sequence $x_n = f_{a_n}^{-1}(y_n)$ is bounded. The sequence $y_n$ is bounded, so it is contained in some interval $(c + \varepsilon, d - \varepsilon)$ for some fixed $\varepsilon > 0$ and $c, d$. By continuity of $f$, for $a$ sufficiently close to $a^*$, we have $f(a, f_{a^*}^{-1}(c))$ must be within $\varepsilon$ of $f(a^*, f_{a^*}^{-1}(c)) = c$. Similarly, $f(a, f_{a^*}^{-1}(d))$ can be within $\varepsilon$ of $d$. Since $f_a^{-1}$ is the inverse of a continuous function, it is monotonic. For $a$ sufficiently close to $a^*$,
    \begin{align*}
        &~ \{y_n\} \text{ is bounded by } c + \varepsilon \text{ and } d - \varepsilon \\
        \Leftrightarrow &~ \{f(a, f_a^{-1}(y_n))\} \text{ is bounded by } c + \varepsilon \text{ and } d - \varepsilon \\
        \Rightarrow &~ \{f(a, f_a^{-1}(y_n))\} \text{ is bounded by } f(a, f_{a^*}^{-1}(c)) \text{ and } f(a, f_{a^*}^{-1}(d)) \\
        \Leftrightarrow &~ \{f_a^{-1}(y_n)\} \text{ is bounded by } f_{a^*}^{-1}(c) \text{ and } f_{a^*}^{-1}(d)
    \end{align*}
    So for sufficiently large $n$, $a_n$ is sufficiently close to $a^*$, and $f_{a_n}^{-1}(y_n)$ is bounded.

    Since the sequence $x_n = f_{a_n}^{-1}(y_n)$ is bounded, it must have a convergent subsequence. Consider any of such convergent subsequence indexed by $n_k$. We have
    \begin{align*}
        f(a^*, \lim_{k \to \infty} x_{n_k}) & = f(\lim_{k \to \infty} a_{n_k}, x_{n_k}) \\
        & = \lim_{k \to \infty} f(a_{n_k}, x_{n_k}) \\
        & = \lim_{k \to \infty} f(a_{n_k}, f_{a_n}^{-1}(y_n)) \\
        & = y^* \\
        & = f(a^*, f_{a^*}^{-1}(y^*)),
    \end{align*}
    and thus $\lim_{k \to \infty} x_{n_k} = f_{a^*}^{-1}(y^*)$. Since $f_{a_n}^{-1}(y_n)$ is bounded and any of its convergent subsequence converges to the same limit $f_{a^*}^{-1}(y^*)$, it also converges to the same limit, completing the proof of continuity.
\end{proof}

\begin{lemma}
\label{lma:op}
    If for any non-decreasing divergent sequence of real numbers $(a_n)_{n \in \NN}$ the sequence of random variables $(X_n)_{n \in \NN}$ is $\bigO_p(a_n)$, then it is also $\bigO_p(1)$.
\end{lemma}
\begin{proof}
    Suppose $X_n \ne \bigO_p(1)$. Then there exists $\varepsilon > 0$ such that for all $M > 0$, there are infinitely many $n$ such that
    \begin{equation}
    \label{eq:not-op1}
        \PP(|X_n| > M) \ge \varepsilon.
    \end{equation}
    Take $n_1$ to be the smallest $n$ satisfying \eqref{eq:not-op1} with $M = 1$. For $i > 1$, take $n_i$ to be the smallest $n > n_{i-1}$ satisfying \eqref{eq:not-op1} with $M = i^2$. Specifically, $(n_i)_{i \in \NN}$ is a strictly increasing sequence such that
        \[
            \PP(|X_{n_i}| > i^2) \ge \varepsilon \quad \text{for all $i$.}
        \]

    Now we are ready to set up a sequence that grows sufficiently slowly to cause a contradiction. For any $n$, take $b_n = i$ where $n \in [n_i, n_{i+1})$. Since $(n_i)_{i \in \NN}$ is strictly increasing and only takes values in integers, $(b_n)_{i \in \NN}$ is a non-decreasing sequence with $\lim_{n \to \infty} b_n = \infty$. So $X_n = \bigO_p(b_n)$ and there exists $M', N'$ such that
        \[
            \PP(|X_n| / b_n > M') < \varepsilon \quad \text{for all $n > N'$.}
        \]
    So for sufficiently large $i > M'$,
        \[
            \varepsilon > \PP(|X_{n_i}| / b_{n_i} > M') = \PP(|X_{n_i}| > M' i) \ge \PP(|X_{n_i}| > i^2) \ge \varepsilon,
        \]
    leading to a contradiction.
\end{proof}

\section{Simulation details}
\label{sec:sim-details}

Each experiment in the Upworthy Research Archive dataset involves two or more treatments corresponding to various combinations of headlines and image ``packages'' associated with an article. The number of impressions and clicks are recorded for each package. The metric of interest is the click-through rate, defined as the ratio of clicks to impressions. We filter out article-package pairs with fewer than $\num{1000}$ impressions or $100$ clicks, to ensure normality approximations are reasonable. For the $\num{4677}$ articles with at least two remaining packages, we arbitrarily consider the one with the most impressions to be the control group and the one with the second-most to be the treatment group, omitting any other packages for that article in the data. From these data, we compute the effect size estimate and the standard error for each experiment. The top-$m$ selection problem is hence selecting the subset of articles, subject to a constraint on the number of articles that can be treated.

For the prior, we applied EB-NN, EB-NSM and EB-NPMLE to the real data. \Cref{fig:delta} shows the density of the unshrunk treatment effects, as well as the observation densities implied by three estimated prior distributions corresponding to three different prior families. EB-NN is clearly misspecified and has thinner tails than the observations, indicating that the distribution of prior effects is not well approximated by a normal distribution. Both EB-NSM and EB-NPMLE result in close fits to the observed data and more realistic tail behavior. As a result, we base our simulation on the more parsimonious model of the two, EB-NSM.

\begin{figure}[htbp]
    \centering
    \includegraphics[width=\columnwidth]{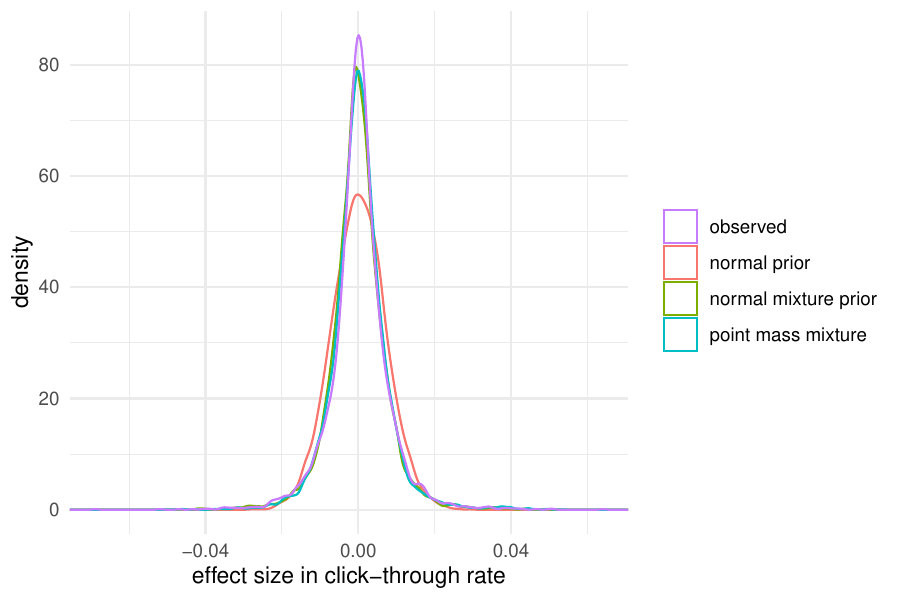}
    \caption{The density of the observed $X$, compared to the densities of observed $X$ as generated by an estimated normal-normal model (EB-NN), an estimated normal scale mixture model (EB-NSM) and a model estimated by NPMLE (EB-NPMLE).}
\label{fig:delta}
\end{figure}

\end{document}